\documentclass[11pt]{article}
\usepackage[utf8]{inputenc}
\pdfoutput=1 
\usepackage{amsmath, amssymb, amsthm, dsfont, mathtools, bm, nicefrac, complexity}
\usepackage[dvipsnames]{xcolor}
\usepackage{hyperref, cleveref, comment}
\usepackage[shortlabels]{enumitem}
\usepackage[ruled,vlined,linesnumbered]{algorithm2e}
\usepackage[left=1.00in, right=1.00in, top=1.00in, bottom=1.00in]{geometry}
\usepackage{subcaption}
\usepackage{tikz}
\usepackage{natbib}
\setcitestyle{authoryear,open={(},close={)}} 

\allowdisplaybreaks

\title{From No-Regret to Strategically Robust Learning in Repeated Auctions}
\author{Junyao Zhao\thanks{IRIF, CNRS, Université Paris Cit\'e. Email: \texttt{junyao@irif.fr}.
}}
\date{}

\newtheorem{theorem}{Theorem}[section]
\newtheorem{lemma}[theorem]{Lemma}
\newtheorem{corollary}[theorem]{Corollary}
\newtheorem{proposition}[theorem]{Proposition}
\newtheorem{claim}[theorem]{Claim}
\newtheorem{definition}[theorem]{Definition}

\newtheorem{example}[theorem]{Example}

\newtheorem{assumption}[theorem]{Assumption}

\def \A {\mathcal{A}}
\def \D {\mathcal{D}}
\def \M {\mathcal{M}}

\def \R {\mathbb{R}}
\def \S {\mathcal{S}}
\def \X {\mathcal{X}}
\def \Z {\mathbb{Z}}
\def \b {\mathbf{b}}
\def \e {\mathbf{e}}
\def \p {\mathbf{p}}
\def \r {\mathbf{r}}

\def \v {\mathbf{v}}
\def \x {\mathbf{x}}

\def \MWU {\textnormal{\sffamily MWU}}
\def \Mye {\textnormal{\sffamily Mye}}
\def \Pair {\textnormal{\sffamily Pair}}
\def \Reg {\textnormal{\sffamily Reg}}
\def \Rev {\textnormal{\sffamily Rev}}
\def \SwapReg {\textnormal{\sffamily SwapReg}}
\let\E\relax
\DeclareMathOperator*{\E}{\mathbb{E}}
\DeclareMathOperator*{\argmax}{arg\,max}
\DeclareMathOperator*{\argmin}{arg\,min}

\begin{document}

\maketitle

\begin{abstract}
In Bayesian single-item auctions, a monotone bidding strategy––one that prescribes a higher bid for a higher value type––can be equivalently represented as a partition of the quantile space into consecutive intervals corresponding to increasing bids.~\citet{KSS24} prove that agile online gradient descent (OGD), when used to update a monotone bidding strategy through its quantile representation, is strategically robust in repeated first-price auctions: when all bidders employ agile OGD in this way, the auctioneer's average revenue per round is at most the revenue of Myerson's optimal auction, regardless of how she adjusts the reserve price over time.

In this work, we show that this strategic robustness guarantee is not unique to agile OGD or to the first-price auction: any no-regret learning algorithm, when fed gradient feedback with respect to the quantile representation, is strategically robust, even if the auction format changes every round, provided the format satisfies allocation monotonicity and voluntary participation. In particular, the multiplicative weights update (MWU) algorithm simultaneously achieves the optimal regret guarantee and a strong strategic robustness guarantee in this auction setting. At a technical level, our results are established via a simple relation that bridges Myerson’s auction theory and standard no-regret learning theory.
\end{abstract}

\section{Introduction}
Consider a repeated single-item auction game over $T$ rounds between an auctioneer and $n$ bidders. In each round, each bidder employs a no-regret learning algorithm\footnote{Informally, an algorithm is no-regret if it performs at least as well as the best fixed strategy in hindsight.}, while the auctioneer, knowing the bidders’ algorithms, strategically chooses the auction format to maximize her cumulative revenue. This setting was first studied by~\citet{BMSW18}, motivated by an influential work of~\citet{NST15} showing that bidder behavior on Microsoft Bing is largely consistent with no-regret learning. This motivation has since been reinforced by the growing adoption of automated bidding agents that employ techniques from no-regret learning~\citep{aggarwal24}.

\citet{BMSW18} showed that many widely adopted no-regret learning algorithms, such as the multiplicative weights update (MWU) algorithm, are susceptible to strategic manipulation: a strategic auctioneer can dynamically choose auction formats to extract the full welfare of a bidder using these algorithms (a result later extended to multiple bidders by~\citet{CWWZ23}), leaving the bidder with zero utility; moreover, the auctioneer can obtain a higher average revenue per round than the revenue of Myerson’s optimal auction by simply adjusting reserve prices in repeated first-price auctions. These results motivated the study of \emph{strategically robust} no-regret learning algorithms––algorithms that, when adopted by all bidders, cap the auctioneer’s average revenue per round at Myerson’s optimal revenue. Importantly, our focus is on the revenue cap in repeated single-item auctions~\citep{BMSW18,KSS24}; we do not address general non-manipulability in arbitrary normal-form or Bayesian games~\citep{DSS19,MMSS22,RZ24,ACMMSS25}.

Most strategically robust (and general non-manipulable) no-regret algorithms designed so far explicitly minimize various notions of swap regret~\citep{BMSW18,MMSS22,ACMMSS25}, which are stronger and more complex than standard regret. An exception is the recent work of~\citet{KSS24}, who showed that a natural no-regret learning algorithm––agile\footnote{Briefly, agile OGD projects an iterate onto the feasible decision space at each update, whereas lazy OGD maintains an unconstrained iterate and projects only at decision time~\citep[Section~5.4.1]{Hazan16}.} online gradient descent (OGD)––is strategically robust in repeated first-price auctions with reserve prices. A key aspect of their result is the choice of decision space:~\citet{KSS24} apply agile OGD over what we call the \emph{quantile strategy space}, which consists of quantile representations of monotone bidding strategies (see Figure~\ref{fig:illustration_quantile_strategy}). Conceptually, their algorithm can be viewed as feeding the gradient of a bidder’s utility with respect to the quantile strategy as the reward vector to agile OGD.

\begin{figure}[ht]
    \centering
    \begin{tikzpicture}[scale=1.2, >=stealth]
        \draw[->] (0,0) -- (4.5,0) node[right] {bidder's value};
        \draw[->] (0,0) -- (0,2.5) node[above] {bid};

        \draw[thick] (0,0) -- (1,0);
        \draw[thick] (1,1) -- (2,1);
        \draw[thick] (2,2) -- (4,2);
        \draw[dashed] (2,1) -- (2,2);
        \draw[dashed] (1,0) -- (1,1);
        \draw[dashed] (2,0) -- (2,1);
        \draw[dashed] (0,2) -- (2,2);
        \draw[dashed] (0,1) -- (1,1);
        \draw[dashed] (4,0) -- (4,2);

        \draw[thick, fill=white] (1,1) circle (2pt);
        \fill (1,0) circle (2pt);
        \draw[thick, fill=white] (2,2) circle (2pt);
        \fill (2,1) circle (2pt);
        \fill (4,2) circle (2pt);
        \draw (0,1) node[left] {$\nicefrac{1}{2}$};
        \draw (0,2) node[left] {$1$};
        \draw (1,0) node[below] {$v_1$};
        \draw (2,0) node[below] {$v_2$};
        \draw (4,0) node[below] {$v_{\max}$};
    \end{tikzpicture}
    \caption{This figure illustrates a monotone (i.e., non-decreasing) bidding strategy that bids $0$ for values in $[0,v_1]$, bids $\frac{1}{2}$ for values in $(v_1,v_2]$, and bids $1$ otherwise. Suppose that $0,\frac{1}{2},1$ are the only possible bids, and let $F$ denote the CDF of the bidder's value distribution.~\citet{KSS24} represent a monotone bidding strategy using the upper-tail probabilities of the value thresholds at which the bid increases ($1-F(v_1)$ and $1-F(v_2)$ in this case). We refer to this quantile-based representation as a quantile strategy. In this paper, we instead use the probability masses of the value intervals corresponding to distinct bids ($F(v_1)$, $F(v_2)-F(v_1)$ and $1-F(v_2)$ in this case) to represent the same quantile strategy, for a reason discussed in Section~\ref{section:technical_overview}.}
    \label{fig:illustration_quantile_strategy}
\end{figure}
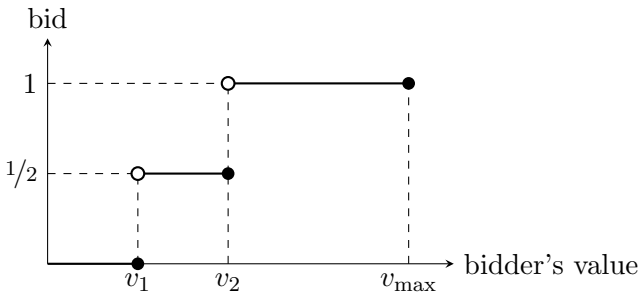

This observation raises a natural structural question that motivates our work:
\begin{quote}
What role does the quantile strategy space itself play in enabling strategic robustness, independent of the specific learning dynamics used?
\end{quote}

In this paper, we show that the quantile strategy space itself plays an important role: strategic robustness can be guaranteed without relying on agile OGD. Specifically, we provide a meta-algorithm (Algorithm~\ref{alg:meta}) that transforms any no-regret learning algorithm into a strategically robust one by feeding it the gradient of a bidder’s utility with respect to the quantile strategy as the reward vector. The resulting strategic robustness guarantee holds even if the auctioneer changes the auction format in every round, provided the format satisfies allocation monotonicity and voluntary participation\footnote{Informally, allocation monotonicity means that raising one’s bid cannot decrease the probability of winning the item, and our voluntary participation condition states that bidders can always ensure zero payment by bidding zero (see Assumption~\ref{assumption:auction_format}). These conditions are satisfied by most auction formats deployed in practice.}.

\begin{theorem}[Informal statement of Theorem~\ref{thm:meta_strategic_robustness} and Proposition~\ref{prop:meta_no_regret}]
If each bidder runs Algorithm~\ref{alg:meta} instantiated with an arbitrary no-regret learning algorithm, and the auction formats chosen by the auctioneer over $T$ rounds satisfy allocation monotonicity and voluntary participation, then the auctioneer’s expected cumulative revenue is at most $\Mye(\D)\cdot T + o(T)$, where $\Mye(\D)$ denotes Myerson’s optimal revenue. Moreover, the bidders incur no regret.
\end{theorem}

In particular, when combined with our meta-algorithm, MWU––despite its known vulnerability to manipulation––achieves optimal regret and strong strategic robustness in our auction setting. The strategic robustness guarantee does not stem from any new property of MWU itself, but from the structural transformation induced by our meta-algorithm.

\begin{corollary}[Informal statement of Corollary~\ref{cor:MWU+meta}]
If all bidders run Algorithm~\ref{alg:meta} instantiated with MWU, and the auction formats chosen by the auctioneer satisfy allocation monotonicity and voluntary participation, then the auctioneer’s expected cumulative revenue is at most $\Mye(\D)\cdot T + O(n\cdot\sqrt{T\log(K)})$, where $K$ is the number of possible bids. Moreover, each bidder’s regret is at most $O(\sqrt{T\log(K)})$.
\end{corollary}

For comparison, the strategic robustness guarantee of the algorithm in~\citet{BMSW18} has a $T^{3/4}$ dependence, while the guarantees of the algorithms in~\citet{MMSS22,KSS24,ACMMSS25}, when applied to the auction setting, exhibit various $K^{\Omega(1)}$ dependencies. However, these algorithms address broader or stronger objectives––such as general non-manipulability in Bayesian and polytope games or incentive compatibility––beyond the revenue cap studied here.

\subsection{Technical overview}\label{section:technical_overview}
Now we provide an informal technical overview that sketches the main ideas behind our results. We first briefly review the analysis in~\citet{KSS24}, and then explain how our approach builds on but also differs from theirs, describing our technical contributions.

\paragraph*{The analysis in~\citet{KSS24}.} \citet{KSS24} study a setting where the auctioneer strategically sets reserve prices in repeated first-price auctions. They first consider the single-bidder case and derive an elegant relation~\citep[page 21]{KSS24} that can be phrased as
\begin{equation}\label{eq:KSS_relation}
    -\nabla u^{(t)}(\bm{p}^{(t)})^{\top}\bm{p}^{(t)}=\textrm{auctioneer's revenue at round $t$}\,-\,\textrm{revenue of a posted-price auction},
\end{equation}
where $\bm{p}^{(t)}$ represents the bidder's quantile strategy (see Figure~\ref{fig:illustration_quantile_strategy} caption) at round $t$, and $u^{(t)}$ is the bidder's utility function at round $t$. Using a potential function argument based on Eq.~\eqref{eq:KSS_relation}, they prove that agile OGD, when used to update the quantile strategy, is strategically robust in the single-bidder case.

For the general setting with multiple bidders, their analysis follows the same plan as~\citet{BMSW18}. Specifically, they first show that if a bidder adopts agile OGD (in a manner similar to the single-bidder case), then in hindsight she can gain only negligible utility by swapping the bidding strategies associated with different value types across the entire time horizon (this property can be viewed as an online analog of incentive compatibility, or as a no-swap-regret property between value types). This, in turn, allows them to construct a global incentive-compatible auction that upper bounds the auctioneer’s average revenue per round. Their proof of the no-swap-regret property is based on another potential function argument inspired by their analysis of the single-bidder case, and it is carried out under a technical assumption that the PDFs of the bidders’ value distributions are bounded everywhere\footnote{Our approach sidesteps this assumption.}.

\paragraph*{Our approach.} The starting point of our approach is the relation in Eq.~\eqref{eq:KSS_relation} from~\citet{KSS24}, and from there we proceed in a different direction. We first identify the right abstraction to extend this relation to the more general setting with multiple bidders, where the auctioneer is allowed to strategically choose auction formats that satisfy allocation monotonicity and voluntary participation. Informally, we construct an incentive-compatible auxiliary auction $\tilde{\M}^{(t)}$ for each round $t\in[T]$ such that (omitting some technical details for clarity)
\begin{equation}\label{eq:general_relation}
    -\sum_{i\in[n]}\nabla u_i^{(t)}(\bm{p}_i^{(t)})^{\top}\bm{p}_i^{(t)}=\textrm{auctioneer's revenue at round $t$}\,-\,\textrm{revenue of auxiliary auction $\tilde{\M}^{(t)}$},
\end{equation}
where $\bm{p}_i^{(t)}$ represents bidder $i$'s quantile strategy at round $t$, and $u_i^{(t)}$ is the bidder $i$'s utility function at round $t$. The derivation of this relation is the main technical contribution of our work.

Then, we prove a simple but crucial fact––if each bidder $i\in[n]$ adopts our meta-algorithm, which feeds the gradient $\nabla u_i^{(t)}(\bm{p}_i^{(t)})$ as the reward vector to a no-regret learning algorithm $\A_i$, then the summation $-\sum_{t\in[T]}\nabla u_i^{(t)}(\bm{p}_i^{(t)})^{\top}\bm{p}_i^{(t)}$ is upper bounded by algorithm $\A_i$'s regret. Thus, by summing both sides of Eq.~\eqref{eq:general_relation} over $T$ rounds, we obtain that
\[
    \sum_{i\in[n]}\textrm{$\A_i$'s regret} \ge \sum_{t\in[T]}\textrm{auctioneer's revenue at round $t$}\,-\,\textrm{revenue of auxiliary auction $\tilde{\M}^{(t)}$}.
\]
Since the auxiliary auction $\tilde{\M}^{(t)}$ is incentive-compatible, its revenue is at most Myerson's optimal revenue. It follows that the auctioneer's total revenue is at most $T$ times Myerson's optimal revenue, plus a lower-order regret term.

Finally, as mentioned in the caption of Figure~\ref{fig:illustration_quantile_strategy}, we use a different representation of quantile strategies in this paper. Instead of using the upper-tail probabilities of the value thresholds at which the bid increases, we represent a quantile strategy by the probability masses of the value intervals that correspond to distinct bids. (Because of this choice, the relation in Eq.~\eqref{eq:general_relation} takes a slightly different form in Lemma~\ref{lem:main_lemma}.) This transforms the decision space into a probability simplex, which enables MWU to recover its optimal regret guarantee and, together with our meta-algorithm, to achieve strong strategic robustness.

\subsection{Related work}
There is a substantial literature on no-regret learning algorithms for repeated auctions~\citep{FGLMS21,BGMMS22,ZHZFW22,BFG23,AFZ25,DLTZ25,HWZ25}. The study of the interaction between a strategic auctioneer and no-regret learning algorithms was initiated by~\citet{BMSW18}, who showed that the auctioneer can extract the full welfare of a bidder adopting any mean-based no-regret algorithm, such as MWU. This result was later extended to multiple bidders by~\citet{CWWZ23}. Similar manipulation phenomena have also been demonstrated in related settings, including principal-agent problems~\citep{GKSTVWW24,LC24}, first-price auctions between two bidders~\citep{RZ24}, and normal-form games~\citep{DSS19}. In the latter case, computing a near-optimal manipulation strategy against MWU is known to be \NP-hard~\citep{ADR25}. Together, these results motivate the design of no-regret learning algorithms that are non-manipulable in general normal-form, Bayesian, and polytope games~\citep{DSS19,MMSS22,ACMMSS25}, or strategically robust in auction environments~\citep{BMSW18,KSS24}.

A prominent line of work achieves general non-manipulability (or strategic robustness in auction environments) by minimizing various notions of swap regret. \citet{BMSW18} design an efficient algorithm for the auction setting with bandit feedback---inspired by the no-swap-regret algorithm of~\citet{BM07}---to minimize a notion of swap regret between value types. For more general Bayesian and polytope games,~\citet{MMSS22} provide an (exponential-time) algorithm based on the no-swap-regret algorithm of~\citet{BM07} to minimize a notion called polytope swap regret, and~\citet{ACMMSS25} design an efficient algorithm to minimize profile swap regret, using Blackwell’s approachability~\citep{Blackwell1956} and semi-separation~\citep{DFFPS25}. In fact, minimizing appropriate notions of swap regret is both sufficient and necessary to achieve non-manipulability in general normal-form games~\citep{DSS19,MMSS22} as well as in Bayesian and polytope games~\citep{MMSS22,RZ24,ACMMSS25}.

Our work is complementary to this literature. Rather than studying general non-manipulability for general games, we focus on a structured auction environment where auction-specific properties can be leveraged to obtain revenue caps. In this context, the recent result of~\citet{KSS24} highlights a distinct phenomenon: they show that agile OGD, when operating over the quantile strategy space, is strategically robust in repeated first-price auctions with reserve prices. Our results identify a general principle underlying this phenomenon: by learning in the quantile strategy space, any no-regret learning algorithm can be transformed into a strategically robust one for repeated single-item auctions, even when the auctioneer is allowed to change the auction format over time, provided the format satisfies allocation monotonicity and voluntary participation.

Finally, we note that there are studies on strategic interactions between users who employ no-regret learning algorithms (e.g.,~\citet{KN22a,KN22b,AHPY24,KHT24}) and between heterogeneous agents~\citep{EKT25}. Other works study additional desirable properties of no-regret learning in strategic environments, such as Pareto optimality~\citep{AECS24,ACMMSS25} and strategic robustness guarantees from the auctioneer’s perspective (e.g.,~\citet{ARS13,MM15,LHW18,DLM19,HK25}).

\paragraph{Concurrent work.} Independently and concurrently, Yang Cai, Haipeng Luo, Chen-Yu Wei, and Weiqiang Zheng (private communication) have shown that any no-regret learning algorithm with gradient feedback under the quantile-based formulation is strategically robust in single-bidder repeated first-price auctions.

\section{Preliminaries}\label{section:preliminaries}
\subsection{Single-item auctions}\label{section:myersons_auction}
In the classic setting of single-item auctions, an auctioneer auctions a single item to $n$ bidders $[n]$. Each bidder $i\in[n]$ has a private value $v_i\in[0,1]$ for the item, drawn independently from a \emph{continuous} prior distribution $\D_i$. We let $\D:=\D_1\times\cdots\times\D_n$ denote the joint distribution of the value vector $\v=(v_1,\dots,v_n)$. For each $i\in[n]$, we let $F_i:[0,1]\to [0,1]$ denote the cumulative distribution function (CDF) of $\D_i$.

The auctioneer's goal is to design an auction, which takes bidders' bids as input and determines who gets the item and how much they pay, in order to maximize her expected revenue. Formally, an auction $\M:=(\x, \p)$ is specified by an \emph{allocation function} $\x:[0,1]^n\to\X([n])$ (where $\X([n]) := \left\{ \mathbf{y} \in \mathbb{R}_{\ge 0}^n \;\middle|\; \sum_{i=1}^n y_i \le 1 \right\}$) and a \emph{payment function} $\p:[0,1]^n\to[0,1]^n$. Specifically, given a bid vector $\b=(b_1,\dots,b_n)$ as input, the $i$-th coordinate $x_i(\b)$ of the output of the allocation function $\x$ is the probability that bidder $i$ receives the item, and the $i$-th coordinate $p_i(\b)$ of the output of the payment function $\p$ is the amount which bidder $i$ pays the auctioneer.

An auction $\M=(\x, \p)$ is \emph{individual-rational} (IR) if each bidder's expected utility is guaranteed to be non-negative when bidding her true value, i.e., for all $i\in[n]$, for any $v_i\in[0,1]$ and $\b_{-i}:=(b_1,\dots,b_{i-1},b_{i+1},\dots,b_n)\in[0,1]^{n-1}$,
\[
    x_i(v_i,\b_{-i})v_i-p_i(v_i,\b_{-i})\ge 0,
\]
where $(v_i,\b_{-i}):=(b_1,\dots,b_{i-1},v_i,b_{i+1},\dots,b_n)$. Moreover, the auction $\M$ is \emph{incentive-compatible} (IC) if each bidder's expected utility is maximized when bidding her true value, i.e., for all $i\in[n]$, for any $v_i,b_i\in[0,1]$ and $\b_{-i}\in[0,1]^{n-1}$,
\[
     x_i(v_i,\b_{-i})v_i-p_i(v_i,\b_{-i})\ge  x_i(b_i,\b_{-i})v_i-p_i(b_i,\b_{-i}).
\]
\citet{myerson81} provides a characterization of IC auctions, a result known as Myerson's lemma\footnote{We note that the original definitions of IC and IR in~\citet{myerson81} are \emph{ex interim}. However, Myerson's optimal auction is \emph{ex post} IC and IR, and Myerson's lemma can be adapted accordingly (see e.g.,~\citet{roughgarden13}).}.
\begin{lemma}[Myerson's lemma]\label{lem:myersons_lemma}
An auction $\M=(\x,\p)$ is IC if and only if the following hold:
\begin{enumerate}[i.]
    \item The allocation function $\x$ is non-decreasing, i.e., for all $i\in[n]$,
    \[
        x_i(b_i,\b_{-i})\le x_i(b_i',\b_{-i}),\,\,\forall\,\b_{-i}\in [0,1]^{n-1},\,\,\forall\,b_i,b_i'\in [0,1] \textrm{ s.t. } b_i\le b_i'.
    \]
    \item For all $i\in[n]$, bidder $i$'s payment $p_i$ satisfies that for any $\b\in[0,1]^n$,
    \[
        p_i(\b)=p_i(0,\b_{-i})+x_i(b_i,\b_{-i})b_i - \int_{0}^{b_i} x_i(z,\b_{-i})dz.
    \]
\end{enumerate}
\end{lemma}
Given a joint distribution $\D$ of bidders' values and an IC and IR auction $\M=(\x, \p)$, we denote the auctioneer's expected revenue by $\Rev_{\M}(\D):=\E_{\v\sim\D}\big[\sum_{i\in[n]}p_i(\v)\big]$. Then, we let $\Mye(\D)$ denote the maximum expected revenue achievable by any IC and IR auction, namely,
\[
    \Mye(\D):=\max_{\M} \Rev_{\M}(\D) \textrm{ s.t.~$\M$ is IC and IR.}
\]
\citet{myerson81} derives an auction that achieves this optimal revenue, known as Myerson’s auction.

\subsection{Repeated auction game with no-regret learners}\label{section:repeated_auction}
In this paper, we study a setting where the auctioneer repeatedly auctions a single item to $n$ bidders. Here, we deviate from Myerson's classic setting, as introduced in Section~\ref{section:myersons_auction}, by assuming that the set of possible bids is finite and equally spaced\footnote{This assumption aligns with practice and prior work~\citep{BMSW18,KSS24}.}:
\begin{assumption}\label{assumption:finite_bids}
The set of possible bids is $B:=\{\frac{i}{K}\mid i\in \{0,\dots,K\}\}$ for some $K\in\Z_{+}$.
\end{assumption}
Since the set of possible bids is finite, in an auction, a bidder needs to choose a \emph{bidding strategy} that specifies a bid for each possible value. We let $\S$ denote the set of all possible bidding strategies, i.e., $\S:=\{s\mid s:[0,1]\to B\}$. Next, we formalize the \emph{repeated auction game} between the auctioneer and $n$ bidders over $T$ rounds:
\begin{enumerate}[(1)]
    \item At each round $t\in[T]$, each bidder $i\in[n]$ selects a bidding strategy $s_i^{(t)}\in\S$.
    \item The auctioneer chooses an auction format $\M^{(t)}=(\x^{(t)},\p^{(t)})$ by specifying an allocation function $\x^{(t)}:B^n\to\X([n])$ and a payment function $\p^{(t)}:B^n\to[0,1]^n$.
    \item Then, a value vector $\v^{(t)}=(v_1^{(t)},\dots,v_n^{(t)})$ is sampled from the prior distribution $\D$, and the corresponding bid vector $\b^{(t)}=(b_1^{(t)},\dots,b_n^{(t)})$ is given by $b_i^{(t)}=s_i^{(t)}(v_i^{(t)})$ for all $i\in[n]$.
    \item At the end of round $t$, the bid vector $\b^{(t)}$ and the auction format $\M^{(t)}$ are revealed to all bidders. Each bidder $i\in[n]$ receives the item with probability $x^{(t)}_i(\b^{(t)})$ and pays $p_i^{(t)}(\b^{(t)})$ to the auctioneer.
\end{enumerate}

We note that, in the repeated auction game, each bidder $i\in[n]$ is effectively solving an online optimization problem: At each round $t\in[T]$, bidder $i$ commits to a bidding strategy $s_i^{(t)}\in \S$, and then a \emph{utility function} $u^{(t)}_i:\S\to[-1,1]$, determined by the auction format $\M^{(t)}$ and other bidders' bids $\b_{-i}^{(t)}=(b^{(t)}_1,\dots,b^{(t)}_{i-1},b^{(t)}_{i+1},\dots,b^{(t)}_n)$, is revealed to her. Specifically, bidder $i$'s utility function $u^{(t)}_i$ is given by
\begin{equation}\label{eq:utility_function}
    u^{(t)}_i(s):=\E_{v_i\sim\D_i}\big[x_i^{(t)}\big(s(v_i),\b_{-i}^{(t)}\big)\cdot v_i-p^{(t)}_i\big(s(v_i),\b_{-i}^{(t)}\big)\big],\,\,\forall\,s\in\S.
\end{equation}
We let $\Reg_i$ denote bidder $i$'s \emph{regret}—the difference between the cumulative utility of the best fixed strategy in hindsight and bidder $i$'s cumulative utility, i.e.,
\begin{equation}\label{eq:regret}
    \Reg_i:=\max_{s\in\S} \sum_{t\in[T]} \big(u^{(t)}_{i}(s) - u^{(t)}_{i}(s_i^{(t)})\big).
\end{equation}
We make the following assumption on bidders' strategic behavior:
\begin{assumption}
Each bidder $i\in[n]$ employs a no-regret learning algorithm $\A_i$. Specifically, at each round $t\in[T]$, the no-regret learning algorithm $\A_i$ takes the utility functions $u_i^{(1)},\dots,u_i^{(t-1)}$ from previous rounds as input and outputs a strategy $s_i^{(t)}$, and it guarantees that $\Reg_i=o(T)$.
\end{assumption}
Moreover, we restrict the class of auction formats available to the auctioneer.
\begin{assumption}\label{assumption:auction_format}
The auction format $\M^{(t)}=(\x^{(t)},\p^{(t)})$ chosen by the auctioneer at each round $t\in[T]$ must satisfy the following conditions:
\begin{enumerate}[i.]
    \item\label{assumption:allocation_monotonicity} Allocation monotonicity: The allocation function $\x^{(t)}$ is non-decreasing, i.e., for all $i\in[n]$,
    \[
        x^{(t)}_i(b_i,\b_{-i})\le x^{(t)}_i(b_i',\b_{-i}),\,\,\forall\,\b_{-i}\in B^{n-1},\,\,\forall\,b_i,b_i'\in B \textrm{ s.t. } b_i\le b_i'.
    \]
    \item\label{assumption:voluntary_participation} Voluntary participation: The payment is zero when a bidder bids zero, i.e., for all $i\in[n]$,
    \[
        p^{(t)}_i(0,\b_{-i})=0,\,\,\forall\,\b_{-i}\in B^{n-1}.
    \]
\end{enumerate}
\end{assumption}
Our goal in this paper is to design no-regret learning algorithms that are \emph{strategically robust}.
\begin{definition}
We say that a no-regret learning algorithm $\A$ is strategically robust in the above repeated auction game if, for any joint prior distribution $\D$, the auctioneer's expected cumulative revenue
\[
    \E_{\v^{(1)},\dots,\v^{(T)}\sim\D}\Big[\sum_{t\in[T]}\sum_{i\in[n]} p_i^{(t)}(\b^{(t)})\Big],
\]
is at most $\Mye(\D)\cdot T+o(T)$ when all bidders adopt $\A$, regardless of the auction formats chosen by the auctioneer (which may be chosen randomly or adaptively, as long as they satisfy Assumption~\ref{assumption:auction_format}).
\end{definition}

\subsection{Discussion on the model assumptions}
We briefly discuss our model assumptions, highlighting differences from the literature.
\begin{proof}[Assumption on auction formats and bidders' feedback]\renewcommand{\qedsymbol}{}
Our setting is in line with previous works by~\citet{BMSW18} and~\citet{KSS24}. The main distinction is the class of auction formats available to the auctioneer. Specifically,~\citet{BMSW18} impose no restrictions on auction formats. \citet{KSS24} focus exclusively on first-price auctions with reserve prices. Our assumption on auction formats (Assumption~\ref{assumption:auction_format}) strikes a middle ground that captures most auctions deployed in practice, including first-price, second-price, and all-pay auctions with reserve prices, posted-price auctions, and Myerson's optimal auction.

Another difference concerns bidders’ feedback. In our general setup, we consider the full-information setting, rather than the more general bandit-feedback setting considered in~\citet{BMSW18}. In our meta-algorithm, it suffices that, at each round $t\in[T]$, each bidder $i\in[n]$ can observe $x_i^{(t)}\big(b,\b_{-i}^{(t)}\big)$ and $p_i^{(t)}\big(b,\b_{-i}^{(t)}\big)$ for all $b\in B$. For example, in a first-price auction, a bidder can infer these quantities given the highest competing bid, which coincides with the feedback setting studied in~\citet{KSS24}.
\end{proof}

\begin{proof}[Randomness from private values]\renewcommand{\qedsymbol}{}
In contrast to the standard no-regret learning literature, where randomization over strategies is typically necessary to guarantee no regret, our repeated auction setting does not require bidders to employ randomization when choosing their bidding strategies. Moreover, the auctioneer selects the auction format after observing bidders' bidding strategies, which seems to make no-regret learning impossible. Nonetheless, we note that there is intrinsic randomness in bidders' private values, which are sampled after the auction format is chosen. This randomness can be used by an algorithm to achieve no regret and strategic robustness, as demonstrated by~\citet{KSS24}.
\end{proof}

\begin{proof}[Consistency with Myerson's setting]\renewcommand{\qedsymbol}{}
As mentioned earlier, our repeated auction setting deviates from Myerson's classic setting by restricting the set of possible bids to a finite set (Assumption~\ref{assumption:finite_bids}), and by allowing essentially arbitrary payment functions (Assumption~\ref{assumption:auction_format}). Therefore, Myerson’s theory cannot be directly applied to the auctions in our setting. However, in our analysis, we will construct auxiliary auctions that enable us to apply Myerson’s lemma in its original form.
\end{proof}

\section{The quantile strategy space}
In this section, we introduce the quantile strategy space, which will be used by our meta-algorithm. We note that our quantile strategy space is a variant of the one\footnote{To be precise, this refers to the decision space of the concave formulation in~\citet[Section 3]{KSS24}, and a related formulation appears in~\citet{KM25}.} used in~\citet{KSS24}, and more generally, quantile-based representations have been widely studied in the auction literature (see e.g.,~\citet[Chapter 3]{Hartline13} and the references therein).

\subsection{Monotone bidding strategies}\label{section:monotone_bidding_strategies}
To motivate the notion of quantile strategies, we first introduce \emph{monotone bidding strategies}.

\begin{definition}\label{def:monotone_bidding_strategy}
We say that a bidding strategy $s:[0,1]\to B$ is monotone if $s(v)\le s(v')$ holds for all $v,v'\in[0,1]$ such that $v\le v'$. We let $\S^{\dagger}$ denote the space of monotone bidding strategies, i.e., $\S^{\dagger}:=\{s:[0,1]\to B\mid s\textrm{ is monotone}\}$.
\end{definition}

For any $i\in[n]$, given a sequence of utility functions $u_i^{(t)}$ and bidder $i$'s bidding strategies $s_i^{(t)}$, we define the regret with respect to the best monotone bidding strategy in hindsight as follows,
\begin{equation}\label{eq:regret_dagger}
    \Reg_i^{\dagger}:=\max_{s\in\S^{\dagger}} \sum_{t\in[T]} \big(u^{(t)}_{i}(s) - u^{(t)}_{i}(s_i^{(t)})\big).
\end{equation}
In Lemma~\ref{lem:monotone_strategies_suffice}, we show that $\Reg_i^{\dagger}$ is equal to the standard regret $\Reg_i$ defined in Eq.~\eqref{eq:regret}, and thus, achieving no regret with respect to the best monotone bidding strategy in hindsight is sufficient to guarantee no regret in general.
\begin{lemma}\label{lem:monotone_strategies_suffice}
Under Assumption~\ref{assumption:auction_format}-\ref{assumption:allocation_monotonicity}, for all $i\in[n]$, it holds that $\Reg_i^{\dagger}=\Reg_i$.
\end{lemma}
To prove Lemma~\ref{lem:monotone_strategies_suffice}, it suffices to show that there exists an optimal bidding strategy that is monotone. Intuitively, when a bidder raises the bid, the resulting change in payment is the same regardless of the bidder's value type. However, the corresponding increase in allocation probability benefits higher value types more than lower ones. Consequently, a higher value type should always bid at least as much as a lower value type to maximize utility. The formal proof is provided in Section~\ref{section:proof_of_lemma_monotone_strategies_suffice}.

\subsection{Quantile strategies and quantile utility functions}\label{section:quantile_space_formulation}
Now we introduce quantile strategies and reformulate each bidder's online optimization problem on the quantile strategy space. Formally, our quantile strategy space is simply a probability simplex.

\begin{definition}\label{def:quantile_strategy}
The quantile strategy space is $\Delta([K+1])$, and any vector $\bm{\pi}\in\Delta([K+1])$ is called a quantile strategy.
\end{definition}

For each bidder $i\in[n]$, we define a corresponding monotone bidding strategy $s_{i,\bm{\pi}}\in\S^{\dagger}$ for each quantile strategy $\bm{\pi}\in\Delta([K+1])$. Conceptually, the monotone bidding strategy $s_{i,\bm{\pi}}$ partitions the value space $[0,1]$ into $K+1$ intervals, where the $j$-th interval corresponds to the bid $\frac{j-1}{K}$ and has probability mass $\pi_j$ under bidder $i$'s CDF.
\begin{definition}\label{def:phi}
For each $i\in[n]$, for any $\bm{\pi}\in\Delta([K+1])$, we define $s_{i,\bm{\pi}}\in\S^{\dagger}$ as follows,
\begin{equation}\label{eq:s_i_pi}
    s_{i,\bm{\pi}}(v)=
    \begin{cases}
    0 & \textrm{if } v\in[0,F_i^{-1}(\pi_1)]  \\
    \frac{j-1}{K} & \textrm{if } v\in(F_i^{-1}(\sum_{\ell=1}^{j-1}\pi_{\ell}),F_i^{-1}(\sum_{\ell=1}^{j}\pi_{\ell})] \textrm{ for } j\in\{2,\dots,K\} \\
    1 & \textrm{if } v\in(F_i^{-1}(\sum_{\ell=1}^{K}\pi_{\ell}),1]
    \end{cases},
\end{equation}
where $F_i$ is the CDF of bidder $i$'s value distribution $\D_i$.
\end{definition}

If a bidder $i\in[n]$ chooses a quantile strategy $\bm{\pi}$ at round $t\in[T]$ of the repeated auction game and bids according to the monotone bidding strategy $s_{i,\bm{\pi}}$, then the resulting utility is\footnote{We note that $q_i^{(t)}$ is independent of $\pi_{K+1}$, since $s_{i,\bm{\pi}}$ does not depend on $\pi_{K+1}$ by Eq.~\eqref{eq:s_i_pi}.}
\begin{equation}\label{eq:quantile_utility_function}
    q_i^{(t)}(\bm{\pi}):=u_i^{(t)}(s_{i,\bm{\pi}}),
\end{equation}
where $u_i^{(t)}$ is defined in Eq.~\eqref{eq:utility_function}. We call $q_i^{(t)}$ the \emph{quantile utility function}. In Lemma~\ref{lem:regret_quantile_strategy}, we reformulate bidder $i$'s regret on the quantile strategy space, using quantile utility functions.

\begin{lemma}\label{lem:regret_quantile_strategy}
For any $i\in[n]$, suppose that bidder $i$ uses the bidding strategy $s_{i,\bm{\pi}^{(t)}}$ corresponding to a quantile strategy $\bm{\pi}^{(t)}\in\Delta([K+1])$ at each round $t\in[T]$. Then, under Assumption~\ref{assumption:auction_format}-\ref{assumption:allocation_monotonicity}, bidder $i$'s regret, as defined in~\eqref{eq:regret}, can be equivalently expressed as
\begin{equation}\label{eq:regret_quantile_strategy}
    \Reg_i=\max_{\bm{\pi}\in\Delta([K+1])} \sum_{t\in[T]}\big(q_i^{(t)}(\bm{\pi})-q_i^{(t)}(\bm{\pi}^{(t)})\big).
\end{equation}
\end{lemma}
Essentially, Lemma \ref{lem:regret_quantile_strategy} follows from Lemma \ref{lem:monotone_strategies_suffice} because a monotone bidding strategy partitions the quantile space into consecutive intervals corresponding to increasing bids, which induces a natural correspondence between monotone bidding strategies and quantile strategies. We defer the formal proof to Section \ref{section:proof_of_lemma_regret_quantile_strategy}.

\subsection{Concavity of quantile utility functions}
By Lemma~\ref{lem:regret_quantile_strategy}, bidders can minimize their regret on the quantile strategy space, as given in Eq.~\eqref{eq:regret_quantile_strategy}, to guarantee no regret in the original sense. In Lemma~\ref{lem:quantile_utility_function_is_concave}, we show that the quantile utility function $q_i^{(t)}$ is concave, which will enable us to apply standard techniques from online convex optimization. Along the way, we also compute the gradient of $q_i^{(t)}$, which will be useful for our analysis later.

\begin{lemma}\label{lem:quantile_utility_function_is_concave}
For any $i\in[n]$ and $t\in[T]$, the quantile utility function $q_i^{(t)}:\Delta([K+1])\to[-1,1]$, as defined in Eq.~\eqref{eq:quantile_utility_function}, is concave. Moreover, we have that $\frac{\partial q_i^{(t)}}{\partial \pi_{K+1}}=0$, and for all $k\in[K]$, the partial derivative of $q_i^{(t)}(\bm{\pi})$ with respect to $\pi_k$ can be expressed as
\begin{equation}\label{eq:gradient_of_quantile_utility_function}
\frac{\partial q_i^{(t)}}{\partial \pi_k}=\Big(\sum_{j=k}^{K}\big(x_i^{(t)}\big(\tfrac{j-1}{K},\b_{-i}^{(t)}\big)-x_i^{(t)}\big(\tfrac{j}{K},\b_{-i}^{(t)}\big)\big)\cdot F_i^{-1}\big(\sum_{\ell=1}^j\pi_{\ell}\big)\Big)+p_i^{(t)}\big(1,\b_{-i}^{(t)}\big)-p^{(t)}_i\big(\tfrac{k-1}{K},\b_{-i}^{(t)}\big),\,\forall\,k\in[K].
\end{equation}
Furthermore, for all $k\in[K+1]$, we have that $\frac{\partial q_i^{(t)}}{\partial \pi_k}\in[-2,1]$ under Assumption~\ref{assumption:auction_format}-\ref{assumption:allocation_monotonicity}.
\end{lemma}
The proof is provided in Section~\ref{section:proof_of_lemma_quantile_utility_function_is_concave}.

\section{From no-regret to strategically robust learning}
In this section, we present our meta-algorithm (Algorithm~\ref{alg:meta}) that converts any standard no-regret learning algorithm into a strategically robust algorithm for our repeated auction game. We begin by formalizing no-regret learning algorithms in the standard online learning setting.
\begin{definition}\label{def:no_regret_learning_algorithm}
In the standard online learning setting, an algorithm $\A$ operates over $T$ rounds. At each round $t\in[T]$, $\A$ chooses a distribution $\bm{\pi}^{(t)}\in\Delta([K+1])$ over $K+1$ possible actions. Then, a reward vector $\r^{(t)}\in[-2,2]^{K+1}$ is revealed, and $\A$ receives the expected utility $\sum_{k\in[K+1]}\pi_k^{(t)}r_k^{(t)}$. The algorithm's regret, which we denote by $\Reg_{\A}$, is the difference between the cumulative reward of the best fixed action in hindsight and the expected cumulative reward of $\A$, i.e.,
\[
    \Reg_{\A}:=\max_{j\in[K+1]} \sum_{t\in[T]} \Big(r^{(t)}_{j} - \sum_{k\in[K+1]}\pi_k^{(t)}r_k^{(t)}\Big).
\]
We say that $\A$ is no-regret if it guarantees that $\Reg_{\A}=o(T)$ for any sequence of reward vectors.
\end{definition}

Algorithm~\ref{alg:meta} runs a no-regret learning algorithm $\A$ over the quantile strategy space. Specifically, at each round $t\in[T]$, Algorithm~\ref{alg:meta} feeds the gradient of the quantile utility function as the reward vector to algorithm $\A$ and uses the distribution returned by $\A$ as a quantile strategy.

\begin{algorithm}[ht]
\SetAlgoLined
\SetKwInOut{Input}{Input}
\SetKwInOut{Output}{Output}
\Input{Online learning algorithm $\A$, and the CDF $F_i$ of bidder $i$'s value distribution $\D_i$}
\SetAlgorithmName{Algorithm}~~
    Initialize algorithm $\A$ with $K+1$ possible actions\;
    \For{$t=1,\dots,T$}{
        Algorithm $\A$ outputs a distribution $\bm{\pi}^{(i,t)}\in\Delta([K+1])$\;
        Let $s_i^{(t)}$ be the monotone bidding strategy $s_{i,\bm{\pi}^{(i,t)}}$ (as defined by Eq.~\eqref{eq:s_i_pi})\;
        Bidder $i$ draws their value $v_i^{(t)}$ from $\D_i$ and bids $b_i^{(t)}=s_i^{(t)}(v_i^{(t)})$\;
        Observe $x_i^{(t)}\big(b,\b_{-i}^{(t)}\big)$ and $p_i^{(t)}\big(b,\b_{-i}^{(t)}\big)$ for all $b\in B$, which (together with $F_i$) are sufficient to compute the gradient $\nabla q_i^{(t)}(\bm{\pi}^{(i,t)})$ (see Lemma~\ref{lem:quantile_utility_function_is_concave})\;
        Feed the gradient $\nabla q_i^{(t)}(\bm{\pi}^{(i,t)})$ as the reward vector $\r^{(t)}$ at round $t$ to algorithm $\A$\;
    }
    \caption{\textsc{Meta-Algorithm (from Bidder $i$'s Perspective)}}
    \label{alg:meta}
\end{algorithm}

Before presenting our main result, we first observe that Algorithm~\ref{alg:meta} guarantees no regret.

\begin{proposition}\label{prop:meta_no_regret}
For any $i\in[n]$, if the input algorithm $\A$ has regret at most $\Reg_{\A}$ in the standard online learning setting, then under Assumption~\ref{assumption:auction_format}-\ref{assumption:allocation_monotonicity}, Algorithm~\ref{alg:meta} guarantees that bidder $i$’s regret $\Reg_i$ (as defined in Eq.~\eqref{eq:regret}) in the repeated auction game is at most $\Reg_{\A}$.
\end{proposition}
\begin{proof}
By Lemma~\ref{lem:quantile_utility_function_is_concave}, the reward vectors $\nabla q_i^{(t)}(\bm{\pi}^{(i,t)})$ provided to the input algorithm $\A$ are in the range $[-2,1]^{K+1}$, which is consistent with the standard online learning setting (Definition~\ref{def:no_regret_learning_algorithm}). Hence, by our assumption in the lemma statement, algorithm $\A$ guarantees that
\begin{align*}
    \Reg_{\A}&\ge\max_{\bm{\pi}\in\Delta([K+1])}\sum_{t\in[T]}\nabla q_i^{(t)}(\bm{\pi}^{(i,t)})^{\top} (\bm{\pi}-\bm{\pi}^{(i,t)})\\
    &\ge\max_{\bm{\pi}\in\Delta([K+1])}\sum_{t\in[T]} q_i^{(t)}(\bm{\pi})-q_i^{(t)}(\bm{\pi}^{(i,t)}) &&\text{(Since $q_i^{(t)}$ is concave by Lemma~\ref{lem:quantile_utility_function_is_concave})}\\
    &=\Reg_i &&\text{(By Lemma~\ref{lem:regret_quantile_strategy})},
\end{align*}
which finishes the proof.
\end{proof}

Now we state our main result, which shows that Algorithm~\ref{alg:meta} is strategically robust.
\begin{theorem}\label{thm:meta_strategic_robustness}
For any joint prior distribution $\D$ of bidders’ values, if each bidder $i\in[n]$ adopts Algorithm~\ref{alg:meta} with an input algorithm $\A_i$ whose regret in the standard online learning setting is at most $\Reg_{\A_i}$, then under Assumption~\ref{assumption:auction_format}, the auctioneer’s expected cumulative revenue is at most $\Mye(\D)\cdot T + \sum_{i\in[n]}\Reg_{\A_i}$.
\end{theorem}

In the remainder of this section, we prove Theorem~\ref{thm:meta_strategic_robustness}. The proof relies on a key lemma that relates the gradient of the quantile utility function to Myerson’s payment rule for a certain auxiliary auction. We first construct these auxiliary auctions in Section~\ref{section:auxiliary_auctions}, then establish the key lemma in Section~\ref{section:main_lemma}, and finally prove Theorem~\ref{thm:meta_strategic_robustness} in Section~\ref{section:proof_of_main_theorem}.

\subsection{Auxiliary auctions}\label{section:auxiliary_auctions}
For each round $t\in [T]$, we construct an auxiliary auction $\tilde{\M}^{(t)}$ as follows.
\begin{definition}\label{def:auxiliary_auction}
For each round $t\in [T]$, given the auction format $\M^{(t)}=(\x^{(t)},\p^{(t)})$ selected by the auctioneer and bidders' monotone bidding strategies $s_i^{(t)}$ chosen by Algorithm~\ref{alg:meta}, the auxiliary auction $\tilde{\M}^{(t)}:=(\tilde{\x}^{(t)},\tilde{\p}^{(t)})$ is specified by an allocation function $\tilde{\x}^{(t)}:[0,1]^n\to\X([n])$ and a payment function $\tilde{\p}^{(t)}:[0,1]^n\to[0,1]^n$. Specifically, the allocation function $\tilde{\x}^{(t)}$ is defined as
\begin{equation}\label{eq:auxiliary_auction_allocation}
    \tilde{\x}^{(t)}(\v):=\x^{(t)}\big(s_1^{(t)}(v_1),\dots,s_n^{(t)}(v_n)\big),\,\,\forall\,\v\in[0,1]^n,
\end{equation}
and the payment function $\tilde{\p}^{(t)}$ is defined by
\begin{equation}\label{eq:auxiliary_auction_payment}
    \tilde{p}_i^{(t)}(\v):=\tilde{x}_i^{(t)}(v_i,\v_{-i})v_i - \int_{0}^{v_i} \tilde{x}_i^{(t)}(z,\v_{-i})dz,\,\,\forall\,\v\in[0,1]^n,\,\,\forall\,i\in[n].
\end{equation}
\end{definition}

Intuitively, the auxiliary auction $\tilde{\M}^{(t)}$ can be viewed as a variant of $\M^{(t)}$ whose allocation function $\x^{(t)}$ is ``stretched'' by the monotone bidding strategies $s_i^{(t)}$ (see Figure~\ref{fig:illustration_auxiliary_auction} for an illustration).

\begin{figure}[ht]
    \centering
    \begin{subfigure}[t]{0.45\textwidth}
        \centering
        \begin{tikzpicture}[scale=1.2, >=stealth]
            \draw[->] (0,0) -- (4,0) node[right] {$b$};
            \draw[->] (0,0) -- (0,2.5) node[above] {$x(b)$};

            \draw[thick] (0,1) -- (1.333,1);
            \draw[thick] (1.333,2) -- (4,2);
            \draw[dashed] (1.333,1) -- (1.333,2);
            \draw[dashed] (1.333,0) -- (1.333,1);
            \draw[dashed] (0,2) -- (1.333,2);

            \fill (1.333,2) circle (2pt);
            \draw[thick, fill=white] (1.333,1) circle (2pt);
            \draw (0,1) node[left] {$\nicefrac{1}{4}$};
            \draw (0,2) node[left] {$\nicefrac{1}{2}$};
            \draw (1.333,0) node[below] {$\nicefrac{1}{3}$};
        \end{tikzpicture}
        \caption{Allocation function $x$}
        \label{fig:allocation_function_x}
    \end{subfigure}
    \hfill
    \begin{subfigure}[t]{0.45\textwidth}
        \centering
        \begin{tikzpicture}[scale=1.2, >=stealth]
            \draw[->] (0,0) -- (4,0) node[right] {$v$};
            \draw[->] (0,0) -- (0,2.5) node[above] {$\tilde{x}(v)$};

            \draw[thick] (0,1) -- (2,1);
            \draw[thick] (2,2) -- (4,2);
            \draw[dashed] (2,1) -- (2,2);
            \draw[dashed] (2,0) -- (2,1);
            \draw[dashed] (0,2) -- (2,2);

            \fill (2,2) circle (2pt);
            \draw[thick, fill=white] (2,1) circle (2pt);
            \draw (0,1) node[left] {$\nicefrac{1}{4}$};
            \draw (0,2) node[left] {$\nicefrac{1}{2}$};
            \draw (2,0) node[below] {$\nicefrac{1}{2}$};
        \end{tikzpicture}
        \caption{Allocation function $\tilde{x}$}
        \label{fig:allocation_function_x_tilde}
    \end{subfigure}
    \caption{This figure illustrates an example of an auxiliary auction $\tilde{\M}$ in the single-bidder case. On the left is the allocation function $x$ for the bidder in the original auction $\M$. On the right is the allocation function $\tilde{x}$ for the bidder in an auxiliary auction $\tilde{\M}$, where $\tilde{x}(v)=x(s(v))$ is obtained by composing $x$ with a monotone bidding strategy $s(v):=\frac{1}{2}\cdot\mathds{1}(v\ge\frac{1}{2})$.}
    \label{fig:illustration_auxiliary_auction}
\end{figure}
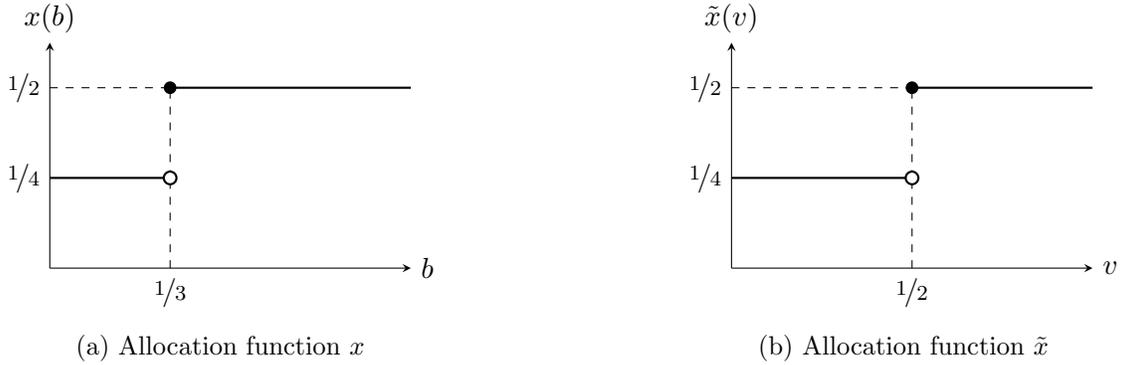

Moreover, we note that $\tilde{\M}^{(t)}$ has a continuous bid space $[0,1]$, which is consistent with Myerson's setting introduced in Section~\ref{section:myersons_auction}. In Claim~\ref{claim:auxiliary_auction_IC_and_IR}, we show that $\tilde{\M}^{(t)}$ is IC and IR.

\begin{claim}\label{claim:auxiliary_auction_IC_and_IR}
Suppose that all bidders adopt Algorithm~\ref{alg:meta} (with arbitrary input algorithms). Then, for each round $t\in [T]$, given any auction format $\M^{(t)}=(\x^{(t)},\p^{(t)})$ (under Assumption~\ref{assumption:auction_format}-\ref{assumption:allocation_monotonicity}) and bidders' monotone bidding strategies $s_i^{(t)}$, the auxiliary auction $\tilde{\M}^{(t)}$ is IC and IR.
\end{claim}
\begin{proof}
We first prove that $\tilde{\M}^{(t)}$ is IC. We notice that the allocation function $\tilde{\x}^{(t)}$ is non-decreasing because it is a composition of the non-decreasing allocation function $\x^{(t)}$ (Assumption~\ref{assumption:auction_format}-\ref{assumption:allocation_monotonicity}) and the monotone bidding strategies $s_i^{(t)}$ generated by Algorithm~\ref{alg:meta}. Moreover, we observe that the payment function $\tilde{\p}^{(t)}$ given by Eq.~\eqref{eq:auxiliary_auction_payment} satisfies Myerson's payment condition in Lemma~\ref{lem:myersons_lemma}. Thus, it follows by Lemma~\ref{lem:myersons_lemma} that $\tilde{\M}^{(t)}$ is IC.

Then, we note that by Eq.~\eqref{eq:auxiliary_auction_payment}, for all $i\in[n]$, for any $v_i\in[0,1]$ and $\b_{-i}\in[0,1]^{n-1}$,
\[
    \tilde{x}_i^{(t)}(v_i,\b_{-i})v_i-\tilde{p}_i^{(t)}(v_i,\b_{-i})=\int_{0}^{v_i} \tilde{x}_i^{(t)}(z,\b_{-i})dz\ge 0,
\]
which shows that $\tilde{\M}^{(t)}$ is IR.
\end{proof}

In Claim~\ref{claim:auxiliary_auction_payment}, we derive a useful expression for $\Rev_{\tilde{\M}^{(t)}}(\D)=\E_{\v\sim\D}\big[\sum_{i\in[n]}\tilde{p}_i^{(t)}(\v)\big]$, namely, the expected revenue obtained by the auxiliary auction $\tilde{\M}^{(t)}$ in the classic auction setting.
\begin{claim}\label{claim:auxiliary_auction_payment}
Suppose that all bidders adopt Algorithm~\ref{alg:meta} (with arbitrary input algorithms). Then, for each round $t\in[T]$, given any auction format $\M^{(t)}$, bidders' quantile strategies $\bm{\pi}^{(i,t)}$, and the corresponding bidding strategies $s_i^{(t)}=s_{i,\bm{\pi}^{(i,t)}}$ (as defined by Eq.~\eqref{eq:s_i_pi}), the expected revenue of the auxiliary auction $\tilde{\M}^{(t)}$ can be expressed as
\[
    \Rev_{\tilde{\M}^{(t)}}(\D)=\sum_{i\in[n]}\E\!_{\v_{-i}^{(t)}}\Big[\sum_{j=1}^K\sum_{k=j+1}^{K+1}\pi_k^{(i,t)}\cdot \Big(x_i^{(t)}\Big(\frac{j}{K},\b_{-i}^{(t)}\Big)-x_i^{(t)}\Big(\frac{j-1}{K},\b_{-i}^{(t)}\Big)\Big)\cdot F_i^{-1}\big(\sum_{\ell=1}^j\pi_{\ell}^{(i,t)}\big)\Big].
\]
\end{claim}
We clarify that in Algorithm~\ref{alg:meta}, the bid vector $\b_{-i}^{(t)}$ is determined by the bidders' quantile/bidding strategies and their private values $\v_{-i}^{(t)}$ at round $t$. Hence, the expectation on the R.H.S.~of the above equation is taken over the randomness of $\v_{-i}^{(t)}$, which induces randomness in $\b_{-i}^{(t)}$.

The technical proof of Claim~\ref{claim:auxiliary_auction_payment} is rather tedious, but it is essentially an application of Myerson's payment rule followed by algebraic manipulations, which we defer to Section~\ref{section:proof_of_claim_auxiliary_auction_payment}.

\subsection{Main lemma}\label{section:main_lemma}
We introduce some notation to simplify the lemma statement. We let $\e:=(1,0,\dots,0)$ denote the first standard basis vector in $\R^{K+1}$. For each round $t\in[T]$, we let $Q^{(t)}:=(\bm{\pi}^{(1,t)},\dots,\bm{\pi}^{(n,t)})$ be the collection of bidders' quantile strategies, which determine their bidding strategies $s_i^{(t)}$, and we let $\Rev_{\M^{(t)}}\big(\D\mid Q^{(t)}\big)$ denote the auctioneer's expected revenue at round $t$ conditioned on $Q^{(t)}$, i.e.,
\begin{equation}\label{eq:revenue_conditioned_on_s_t}
    \Rev_{\M^{(t)}}\big(\D\mid Q^{(t)}\big):=\sum_{i\in[n]}\E_{\v^{(t)}\sim\D}\big[p_i^{(t)}\big(s_1^{(t)}(v_1^{(t)}),\dots,s_n^{(t)}(v_n^{(t)})\big) \mid Q^{(t)}\big].
\end{equation}

Now we state the main lemma (Lemma~\ref{lem:main_lemma}), which relates the gradient of the quantile utility function to the difference between the auctioneer's expected revenue at round $t$ and the expected revenue of the auxiliary auction. A similar relation for the special case of single-bidder first-price auctions was established by~\citet[page 21]{KSS24}.
\begin{lemma}\label{lem:main_lemma}
Suppose that all bidders adopt Algorithm~\ref{alg:meta} (with arbitrary input algorithms). Then, for each round $t\in[T]$, given any auction format $\M^{(t)}$ (under Assumption~\ref{assumption:auction_format}-\ref{assumption:voluntary_participation}), bidders' quantile strategies $\bm{\pi}^{(i,t)}$, and the corresponding bidding strategies $s_i^{(t)}=s_{i,\bm{\pi}^{(i,t)}}$ (as defined by Eq.~\eqref{eq:s_i_pi}), we have that
\begin{equation}\label{eq:main_lemma}
    \sum_{i\in[n]} \E\!_{\v_{-i}^{(t)}}\big[\nabla q_i^{(t)}(\bm{\pi}^{(i,t)})^{\top}(\e-\bm{\pi}^{(i,t)})\mid Q^{(t)}\big]=\Rev_{\M^{(t)}}\big(\D \mid Q^{(t)}\big)-\Rev_{\tilde{\M}^{(t)}}(\D).
\end{equation}
\end{lemma}
We clarify that in Algorithm~\ref{alg:meta}, bidder $i$'s quantile utility function $q_i^{(t)}$ at round $t$ (as defined in Eq.~\eqref{eq:quantile_utility_function}) depends on the other bidders' bids $\b_{-i}^{(t)}$, which in turn are determined by the other bidders' quantile/bidding strategies and private values $\v_{-i}^{(t)}$ at round $t$. Hence, the expectation in Eq.~\eqref{eq:main_lemma} is taken over the randomness of $\v_{-i}^{(t)}$. We proceed to the proof of Lemma~\ref{lem:main_lemma}.
\begin{proof}[Proof of Lemma~\ref{lem:main_lemma}]
For all $i\in[n]$ and $\b_{-i}^{(t)}\in B^{n-1}$, we denote
\begin{equation}\label{eq:delta_i_j}
    \delta_{i,j}^{(t)}(\b_{-i}^{(t)}):=\big(x_i^{(t)}\big(\tfrac{j-1}{K},\b_{-i}^{(t)}\big)-x_i^{(t)}\big(\tfrac{j}{K},\b_{-i}^{(t)}\big)\big)\cdot F_i^{-1}\big(\sum_{\ell=1}^j \pi_{\ell}^{(i,t)}\big),\,\,\forall\,j\in[K].
\end{equation}
Then, by Eq.~\eqref{eq:gradient_of_quantile_utility_function} in Lemma~\ref{lem:quantile_utility_function_is_concave}, we can succinctly represent the partial derivatives $\frac{\partial q_i^{(t)}}{\partial \pi_k^{(i,t)}}$ as
\begin{equation}\label{eq:gradient_of_quantile_utility_function_compact}
\frac{\partial q_i^{(t)}}{\partial \pi_k^{(i,t)}}=\Big(\sum_{j=k}^{K}\delta_{i,j}^{(t)}(\b_{-i}^{(t)})\Big)+p^{(t)}_i\big(1,\b_{-i}^{(t)}\big)-p^{(t)}_i\big(\tfrac{k-1}{K},\b_{-i}^{(t)}\big),\,\,\forall\,k\in[K].
\end{equation}
\subsubsection*{Step 1: Decomposing $\nabla q_i^{(t)}(\bm{\pi}^{(i,t)})^{\top}(\e-\bm{\pi}^{(i,t)})$}
We first expand $\nabla q_i^{(t)}(\bm{\pi}^{(i,t)})^{\top}(\e-\bm{\pi}^{(i,t)})$ as follows,
\begin{align}\label{eq:expand_dot_product}
    &\nabla q_i^{(t)}(\bm{\pi}^{(i,t)})^{\top}(\e-\bm{\pi}^{(i,t)})\nonumber\\
    =&\,\frac{\partial q_i^{(t)}}{\partial \pi_1^{(i,t)}}-\sum_{k=1}^{K+1}\pi_k^{(i,t)}\cdot\frac{\partial q_i^{(t)}}{\partial \pi_k^{(i,t)}}\nonumber\\
    =&\,\frac{\partial q_i^{(t)}}{\partial \pi_1^{(i,t)}}-\sum_{k=1}^{K}\pi_k^{(i,t)}\cdot\frac{\partial q_i^{(t)}}{\partial \pi_k^{(i,t)}}\qquad\qquad\qquad\qquad\qquad\qquad\qquad\,\,\text{(Since $\frac{\partial q_i^{(t)}}{\partial \pi_{K+1}^{(i,t)}}=0$ by Lemma~\ref{lem:quantile_utility_function_is_concave})}\nonumber\\
    =&\Big(\sum_{j=1}^{K}\delta_{i,j}^{(t)}(\b_{-i}^{(t)})-p^{(t)}_i\big(0,\b_{-i}^{(t)}\big)\Big)-\sum_{k=1}^{K}\pi_k^{(i,t)}\cdot\Big(\sum_{j=k}^{K}\delta_{i,j}^{(t)}(\b_{-i}^{(t)})-p^{(t)}_i\big(\tfrac{k-1}{K},\b_{-i}^{(t)}\big)\Big)\nonumber\\
    &+\pi_{K+1}^{(i,t)}\cdot p^{(t)}_i\big(1,\b_{-i}^{(t)}\big)\qquad\qquad\qquad\qquad\qquad\qquad\qquad\qquad\quad\,\,\,\,\text{(By Eq.~\eqref{eq:gradient_of_quantile_utility_function_compact} and $\textstyle\sum\limits_{k=1}^{K+1}\pi_k^{(i,t)}=1$)}\nonumber\\
    =&\Big(\sum_{j=1}^{K}\delta_{i,j}^{(t)}(\b_{-i}^{(t)})\Big)-\sum_{k=1}^{K}\pi_k^{(i,t)}\cdot\Big(\sum_{j=k}^{K}\delta_{i,j}^{(t)}(\b_{-i}^{(t)})-p^{(t)}_i\big(\tfrac{k-1}{K},\b_{-i}^{(t)}\big)\Big)\nonumber\\
    &+\pi_{K+1}^{(i,t)}\cdot p^{(t)}_i\big(1,\b_{-i}^{(t)}\big) \qquad\qquad\qquad\qquad\qquad\qquad\qquad\qquad\qquad\qquad\quad\,\,\text{(By Assumption~\ref{assumption:auction_format}-\ref{assumption:voluntary_participation})}\nonumber\\
    =&\Big(\sum_{j=1}^{K}\Big(1-\sum_{k=1}^j\pi_k^{(i,t)}\Big)\cdot\delta_{i,j}^{(t)}(\b_{-i}^{(t)})\Big)+\sum_{k=1}^{K+1}\pi_k^{(i,t)}\cdot p^{(t)}_i\big(\tfrac{k-1}{K},\b_{-i}^{(t)}\big)\qquad\quad\,\,\,\,\text{(By exchanging the sums)}\nonumber\\
    =&\Big(\sum_{j=1}^{K}\sum_{k=j+1}^{K+1}\pi_k^{(i,t)}\cdot\delta_{i,j}^{(t)}(\b_{-i}^{(t)})\Big)+\sum_{k=1}^{K+1}\pi_k^{(i,t)}\cdot p^{(t)}_i\big(\tfrac{k-1}{K},\b_{-i}^{(t)}\big) \qquad\qquad\qquad\qquad\text{(Since $\textstyle\sum\limits_{k=1}^{K+1}\pi_k^{(i,t)}=1$)}.
\end{align}

\subsubsection*{Step 2: Connecting the first term to $\Rev_{\tilde{\M}^{(t)}}(\D)$}
We observe that by Claim~\ref{claim:auxiliary_auction_payment} and Eq.~\eqref{eq:delta_i_j},
\begin{equation}\label{eq:first_term_and_M_tilde}
    \Rev_{\tilde{\M}^{(t)}}(\D)=-\sum_{i\in[n]}\E\!_{\v_{-i}^{(t)}}\Big[\sum_{j=1}^{K}\sum_{k=j+1}^{K+1}\pi_k^{(i,t)}\cdot\delta_{i,j}^{(t)}(\b_{-i}^{(t)})\Big].
\end{equation}

\subsubsection*{Step 3: Identifying the second term with $\Rev_{\M^{(t)}}\big(\D \mid Q^{(t)}\big)$}
Moreover, we notice that for any $i\in[n]$ and given any $\b_{-i}^{(t)}\in B^{n-1}$,
\begin{align}\label{eq:dot_product_second_term}
    \E_{v_i^{(t)}\sim\D_i} \big[p^{(t)}_i\big(s_i^{(t)}(v_i^{(t)}),\b_{-i}^{(t)}\big)\big]&=\E_{v_i^{(t)}\sim\D_i} \big[p^{(t)}_i\big(s_{i,\bm{\pi}^{(i,t)}}(v_i^{(t)}),\b_{-i}^{(t)}\big)\big]&&\text{(Since $s_i^{(t)}=s_{i,\bm{\pi}^{(i,t)}}$)}\nonumber\\
    &=\sum_{k=1}^{K+1}\pi_k^{(i,t)}\cdot p^{(t)}_i\big(\tfrac{k-1}{K},\b_{-i}^{(t)}\big)&&\text{(By Eq.~\eqref{eq:s_i_pi})}.
\end{align}
Hence, we have that
\begin{align}\label{eq:second_term_and_M}
\Rev_{\M^{(t)}}\big(\D \mid Q^{(t)}\big)&=\sum_{i\in[n]}\E_{\v^{(t)}\sim\D}\big[p^{(t)}_i\big(s_1^{(t)}(v_1^{(t)}),\dots,s_n^{(t)}(v_n^{(t)})\big)\mid Q^{(t)}\big] &&\text{(By Eq.~\eqref{eq:revenue_conditioned_on_s_t})}\nonumber\\
&=\sum_{i\in[n]}\E\!_{\v_{-i}^{(t)}}\Big[\E_{v_i^{(t)}\sim\D_i} \big[p^{(t)}_i\big(s_i^{(t)}(v_i^{(t)}),\b_{-i}^{(t)}\big)\big]\mid Q^{(t)}\Big]&&\text{(Since $b_{j}^{(t)}=s_{j}^{(t)}(v_{j}^{(t)})$)}\nonumber\\
&=\sum_{i\in[n]}\E\!_{\v_{-i}^{(t)}}\Big[\sum_{k=1}^{K+1}\pi_k^{(i,t)}\cdot p^{(t)}_i\big(\tfrac{k-1}{K},\b_{-i}^{(t)}\big)\mid Q^{(t)}\Big]&&\text{(By Eq.~\eqref{eq:dot_product_second_term})}.
\end{align}

\subsubsection*{Step 4: Combining the two terms}
Putting together Eq.~\eqref{eq:expand_dot_product},~\eqref{eq:first_term_and_M_tilde} and~\eqref{eq:second_term_and_M}, we obtain that
\begin{align*}
    &\sum_{i\in[n]}\E\!_{\v_{-i}^{(t)}}\big[\nabla q_i^{(t)}(\bm{\pi}^{(i,t)})^{\top}(\e-\bm{\pi}^{(i,t)})\mid Q^{(t)}\big]\\
    =&\sum_{i\in[n]}\E\!_{\v_{-i}^{(t)}}\Big[\sum_{j=1}^{K}\sum_{k=j+1}^{K+1}\pi_k^{(i,t)}\cdot\delta_{i,j}^{(t)}(\b_{-i}^{(t)})\mid Q^{(t)}\Big]+\sum_{i\in[n]}\E\!_{\v_{-i}^{(t)}}\Big[\sum_{k=1}^{K+1}\pi_k^{(i,t)}\cdot p^{(t)}_i\big(\tfrac{k-1}{K},\b_{-i}^{(t)}\big)\mid Q^{(t)}\Big]\\
    =&\,\Rev_{\M^{(t)}}\big(\D \mid Q^{(t)}\big)-\Rev_{\tilde{\M}^{(t)}}(\D),
\end{align*}
which establishes the lemma.
\end{proof}

\subsection{Proof of the main theorem}\label{section:proof_of_main_theorem}
We are now ready to prove Theorem~\ref{thm:meta_strategic_robustness}.
\begin{proof}[Proof of Theorem~\ref{thm:meta_strategic_robustness}]
Let $\Rev_{acc}$ denote the auctioneer's expected accumulated revenue over all $T$ rounds. Recall that for each $t\in[T]$, $\Rev_{\M^{(t)}}\big(\D\mid Q^{(t)}\big)$ (defined in Eq.~\eqref{eq:revenue_conditioned_on_s_t}) is the auctioneer's expected revenue at round $t$ conditioned on $Q^{(t)}$. Hence, we have that
\begin{equation}\label{eq:accumulated_revenue}
    \Rev_{acc}=\sum_{t\in[T]}\E\!_{Q^{(t)}}\big[\Rev_{\M^{(t)}}\big(\D \mid Q^{(t)}\big)\big].
\end{equation}
Recall that $\Mye(\D)$ denotes the expected revenue of the optimal auction for $\D$. We derive that
\begin{align}\label{eq:strategic_loss}
    \Rev_{acc}-\Mye(\D)\cdot T&=\sum_{t\in[T]}\E\!_{Q^{(t)}}\big[\Rev_{\M^{(t)}}\big(\D \mid Q^{(t)}\big)-\Mye(\D)\big]\qquad\qquad\qquad\qquad\,\,\,\text{(By Eq.~\eqref{eq:accumulated_revenue})}\nonumber\\
    &\le\sum_{t\in[T]}\E\!_{Q^{(t)}}\big[\Rev_{\M^{(t)}}\big(\D \mid Q^{(t)}\big)-\Rev_{\tilde{\M}^{(t)}}(\D)\big]\qquad\qquad\qquad\,\,\,\,\,\text{(By Claim~\ref{claim:auxiliary_auction_IC_and_IR})}\nonumber\\
    &=\sum_{t\in[T]}\E\!_{Q^{(t)}}\Big[\sum_{i\in[n]}\E\!_{\v_{-i}^{(t)}}\big[\nabla q_i^{(t)}(\bm{\pi}^{(i,t)})^{\top}(\e-\bm{\pi}^{(i,t)})\mid Q^{(t)}\big]\Big]\qquad\,\,\,\,\text{(By Lemma~\ref{lem:main_lemma})}\nonumber\\
    &=\sum_{t\in[T]}\E\!_{Q^{(t)}}\Big[\sum_{i\in[n]}\E\!_{\v^{(t)}}\big[\nabla q_i^{(t)}(\bm{\pi}^{(i,t)})^{\top}(\e-\bm{\pi}^{(i,t)})\mid Q^{(t)}\big]\Big]\nonumber\\
    &=\sum_{t\in[T]}\sum_{i\in[n]}\E\!_{Q^{(t)},\v^{(t)}}\big[\nabla q_i^{(t)}(\bm{\pi}^{(i,t)})^{\top}(\e-\bm{\pi}^{(i,t)})\big]\nonumber\\
    &=\sum_{t\in[T]}\sum_{i\in[n]}\E\!_{Q^{(1)},\dots,Q^{(T)},\v^{(1)},\dots,\v^{(T)}}\big[\nabla q_i^{(t)}(\bm{\pi}^{(i,t)})^{\top}(\e-\bm{\pi}^{(i,t)})\big]\nonumber\\
    &=\sum_{i\in[n]}\E\!_{Q^{(1)},\dots,Q^{(T)},\v^{(1)},\dots,\v^{(T)}}\Big[\sum_{t\in[T]}\nabla q_i^{(t)}(\bm{\pi}^{(i,t)})^{\top}(\e-\bm{\pi}^{(i,t)})\Big].
\end{align}
Notice that, by Lemma~\ref{lem:quantile_utility_function_is_concave}, for each $i\in[n]$, the reward vectors $\nabla q_i^{(t)}(\bm{\pi}^{(i,t)})$ provided to bidder $i$'s input algorithm $\A_i$ are in the range $[-2,1]^{K+1}$, which is consistent with the standard online learning setting (Definition~\ref{def:no_regret_learning_algorithm}). Hence, by our assumption in the lemma statement, for each $i\in[n]$, algorithm $\A_i$ guarantees that
\[
    \max_{\bm{\pi}\in\Delta([K+1])}\sum_{t\in[T]}\nabla q_i^{(t)}(\bm{\pi}^{(i,t)})^{\top} (\bm{\pi}-\bm{\pi}^{(i,t)})\le\Reg_{\A_i}.
\]
In particular, since $\e\in\Delta([K+1])$, it follows from the above inequality that
\[
    \sum_{t\in[T]}\nabla q_i^{(t)}(\bm{\pi}^{(i,t)})^{\top} (\e-\bm{\pi}^{(i,t)})\le\Reg_{\A_i}.
\]
This together with Ineq.~\eqref{eq:strategic_loss} implies that $\Rev_{acc}-\Mye(\D)\cdot T\le\sum_{i\in[n]}\Reg_{\A_i}$.
\end{proof}

\section{Instantiating the meta-algorithm with MWU}
In this section, we analyze the behavior of the MWU algorithm when instantiated within our meta-algorithm (Algorithm~\ref{alg:meta}). While MWU is known to be manipulable~\citep{BMSW18}, we show that it achieves the optimal regret and the best-known strategic robustness guarantee in our auction setting, when combined\footnote{For readers familiar with online convex optimization, this combination is equivalent to applying the exponentiated gradient algorithm~\citep[Algorithm 15]{Hazan16} to our quantile utility functions.} with our meta-algorithm (Algorithm~\ref{alg:meta}). Importantly, the strategic robustness guarantee we obtain does not stem from any new property of MWU itself, but from the structural transformation induced by our meta-algorithm.

We refer interested readers to~\citet{AHK12} for a comprehensive exposition of the MWU algorithm. Here, we simply state its regret guarantee in Lemma~\ref{lem:MWU}.

\begin{lemma}[{See e.g.,~\citet[Theorem 2.1]{AHK12}}]\label{lem:MWU}
Given any sequence of reward vectors $\r^{(1)},\dots,\r^{(T)}\in[-2,2]^{K+1}$, the MWU algorithm guarantees that $\Reg_{\MWU}=O(\sqrt{T\log(K)})$.
\end{lemma}

In Corollary~\ref{cor:MWU+meta}, we establish the regret and strategic robustness guarantees of Algorithm~\ref{alg:meta} when it is instantiated with the MWU algorithm. The regret guarantee follows from Proposition~\ref{prop:meta_no_regret} and Lemma~\ref{lem:MWU}, and the strategic robustness guarantee follows from Theorem~\ref{thm:meta_strategic_robustness} and Lemma~\ref{lem:MWU}.
\begin{corollary}\label{cor:MWU+meta}
Suppose that all bidders adopt Algorithm~\ref{alg:meta} with the MWU algorithm as the input algorithm. Then, for any joint prior distribution $\D$ of bidders’ values, under Assumption~\ref{assumption:auction_format}, the following hold:
\begin{enumerate}[i.]
    \item The regret $\Reg_i$ of each bidder $i\in[n]$ (as defined in Eq.~\eqref{eq:regret}) is at most $O(\sqrt{T\log(K)})$.
    \item The auctioneer’s expected cumulative revenue is at most $\Mye(\D)\cdot T + O(n\cdot\sqrt{T\log(K)})$.
\end{enumerate}
\end{corollary}
For comparison, the standard regret guarantee of the agile and lazy online gradient descent (OGD) algorithms~\citep[Section 5.4.1]{Hazan16} is $O(\sqrt{T\cdot K})$. Consequently, combining the agile or lazy OGD algorithm with Algorithm~\ref{alg:meta} results in regret and strategic‐robustness guarantees with $\sqrt{T\cdot K}$ dependence. This matches the regret guarantee established by~\citet{KSS24} for agile OGD in first-price auctions, and slightly improves their strategic robustness guarantee for multiple bidders, which scales linearly with $K$. We note that for multiple bidders,~\citet{KSS24} analyze a variant of agile OGD designed for unknown prior distributions, and their analysis establishes a stronger incentive compatibility property under a technical assumption.

Finally, in Proposition~\ref{prop:lower_bound}, we show that the regret guarantee in Corollary~\ref{cor:MWU+meta} is optimal for our repeated auction game, by reducing to a well-known hard instance from the online learning literature~\citep[Theorem 3.7]{CG06}. The proof is provided in Section~\ref{section:proof_of_prop_lower_bound}.
\begin{proposition}\label{prop:lower_bound}
In the repeated auction game under Assumption~\ref{assumption:auction_format}, a bidder's regret (as defined in Eq.~\eqref{eq:regret}) is $\Omega(\sqrt{T\log(K)})$ in the worst case, regardless of the algorithm used.
\end{proposition}

\section{Discussion}
We have shown that the quantile strategy space can be used to transform any no-regret learning algorithm into a strategically robust one, for repeated auctions satisfying allocation monotonicity and voluntary participation. It is, however, still unclear whether agile OGD itself is also strategically robust in repeated first-price auctions without relying on the quantile strategy space. Specifically, suppose all bidders run a separate instance of agile OGD for each possible value type\footnote{Although the value range is a continuous interval $[0,1]$, we can discretize it into sufficiently small sub-intervals and select a representative value type from each.} and use that instance to determine bids for that value type in repeated first-price auctions. Can the auctioneer strategically choose reserve prices to obtain a total revenue of $\Mye(\D)\cdot T+\Omega(T)$ when facing this algorithm? We did not resolve this question in this paper. However, we note that agile OGD can incur high swap regret (Section~\ref{section:swap_regret}), and therefore, it is manipulable in repeated normal-form games~\citep[Theorem 3]{MMSS22}.

\section*{Acknowledgments}
The author is supported by a postdoctoral fellowship from the Fondation Sciences Mathématiques de Paris (FSMP). The author thanks Jon Schneider for valuable feedback, Weiqiang Zheng for discussions on related independent and concurrent work, and Adrian Vladu and Omri Weinstein for discussions on swap regret and online learning.

\bibliography{cite}

\appendix
\section{Supplementary proofs}
\subsection{Proof of Lemma~\ref{lem:monotone_strategies_suffice}}\label{section:proof_of_lemma_monotone_strategies_suffice}
\begin{proof}[Proof of Lemma~\ref{lem:monotone_strategies_suffice}]
It suffices to show that $\max_{s\in\S} \sum_{t\in[T]} u^{(t)}_{i}(s)=\max_{s\in\S^{\dagger}} \sum_{t\in[T]} u^{(t)}_{i}(s)$. First, we derive that
\begin{align*}
    \max_{s\in\S}\sum_{t\in[T]} u^{(t)}_{i}(s)&=\max_{s\in\S}\sum_{t\in[T]}\E_{v_i\sim\D_i}\big[x_i^{(t)}\big(s(v_i),\b_{-i}^{(t)}\big)\cdot v_i-p^{(t)}_i\big(s(v_i),\b_{-i}^{(t)}\big)\big] &&\text{(By Eq.~\eqref{eq:utility_function})}\\
    &=\max_{s\in\S}\E_{v_i\sim\D_i}\Big[\sum_{t\in[T]}x_i^{(t)}\big(s(v_i),\b_{-i}^{(t)}\big)\cdot v_i-p^{(t)}_i\big(s(v_i),\b_{-i}^{(t)}\big)\Big]\\
    &=\E_{v_i\sim\D_i}\Big[\max_{s(v_i)\in B}\sum_{t\in[T]}x_i^{(t)}\big(s(v_i),\b_{-i}^{(t)}\big)\cdot v_i-p^{(t)}_i\big(s(v_i),\b_{-i}^{(t)}\big)\Big].
\end{align*}
The bidding strategy $s^*$ that maximizes the above expectation is given by
\[
    s^*(v_i):=\argmax_{s(v_i)\in B} \sum_{t\in[T]}x_i^{(t)}\big(s(v_i),\b_{-i}^{(t)}\big)\cdot v_i-p^{(t)}_i\big(s(v_i),\b_{-i}^{(t)}\big),\,\,\forall\,v_i\in[0,1],
\]
where we choose $s^*(v_i)$ to be the largest maximizer in case of ties. Now we prove that $s^*(v_i)\ge s^*(v_i')$ for any $v_i,v_i'\in[0,1]$ such that $v_i>v_i'$, which implies that $s^*\in\S^{\dagger}$ and establishes the lemma.

To this end, suppose for contradiction that $s^*(v_i)< s^*(v_i')$. Then, the pair
$$\Pair_i:=\big(\sum_{t\in[T]}x_i^{(t)}\big(s^*(v_i),\b_{-i}^{(t)}\big),\sum_{t\in[T]}p^{(t)}_i\big(s^*(v_i),\b_{-i}^{(t)}\big)\big)$$
must differ from the pair
$$\Pair_i':=\big(\sum_{t\in[T]}x_i^{(t)}\big(s^*(v_i'),\b_{-i}^{(t)}\big),\sum_{t\in[T]}p^{(t)}_i\big(s^*(v_i'
),\b_{-i}^{(t)}\big)\big),$$
because otherwise $s^*(v_i)$ would not be the largest maximizer, which contradicts our choice of $s^*(v_i)$. Next, we consider the following two cases and derive contradictions in both.
\subsubsection*{Case 1: $\sum_{t\in[T]}x_i^{(t)}\big(s^*(v_i),\b_{-i}^{(t)}\big)=\sum_{t\in[T]}x_i^{(t)}\big(s^*(v_i'),\b_{-i}^{(t)}\big)$}
In this case, both bids $s^*(v_i)$ and $s^*(v_i')$ induce the same cumulative allocation. Moreover, since $\Pair_i\neq \Pair_i'$, we must have that
$$\sum_{t\in[T]}p^{(t)}_i\big(s^*(v_i),\b_{-i}^{(t)}\big)\neq \sum_{t\in[T]}p^{(t)}_i\big(s^*(v_i'
),\b_{-i}^{(t)}\big).$$
This means that one of the two bids, $s^*(v_i)$ or $s^*(v_i')$, incurs a strictly higher cumulative payment, even though both induce the same cumulative allocation. Hence, the bid that incurs the higher payment is strictly suboptimal for both value types $v_i$ and $v_i'$, which contradicts our choice of $s^*(v_i)$ and $s^*(v_i')$.

\subsubsection*{Case 2: $\sum_{t\in[T]}x_i^{(t)}\big(s^*(v_i),\b_{-i}^{(t)}\big)\neq\sum_{t\in[T]}x_i^{(t)}\big(s^*(v_i'),\b_{-i}^{(t)}\big)$}
In this case, we have that
$$\sum_{t\in[T]}x_i^{(t)}\big(s^*(v_i),\b_{-i}^{(t)}\big)<\sum_{t\in[T]}x_i^{(t)}\big(s^*(v_i'),\b_{-i}^{(t)}\big),$$
since we assume that $s^*(v_i)< s^*(v_i')$, and that each $x_i^{(t)}$ is non-decreasing (Assumption~\ref{assumption:auction_format}-\ref{assumption:allocation_monotonicity}). Moreover, it follows from the definition of $s^*(v_i)$ that
\[
    \sum_{t\in[T]}x_i^{(t)}\big(s^*(v_i),\b_{-i}^{(t)}\big)\cdot v_i-p^{(t)}_i\big(s^*(v_i),\b_{-i}^{(t)}\big)\ge\sum_{t\in[T]}x_i^{(t)}\big(s^*(v_i'),\b_{-i}^{(t)}\big)\cdot v_i-p^{(t)}_i\big(s^*(v_i'),\b_{-i}^{(t)}\big).
\]
By rearranging, this is equivalent to
\[
    \sum_{t\in[T]}\big(x_i^{(t)}\big(s^*(v_i'),\b_{-i}^{(t)}\big)-x_i^{(t)}\big(s^*(v_i),\b_{-i}^{(t)}\big)\big)\cdot v_i\le\sum_{t\in[T]}p^{(t)}_i\big(s^*(v_i'),\b_{-i}^{(t)}\big)-p^{(t)}_i\big(s^*(v_i),\b_{-i}^{(t)}\big).
\]
Since $\sum_{t\in[T]}x_i^{(t)}\big(s^*(v_i),\b_{-i}^{(t)}\big)<\sum_{t\in[T]}x_i^{(t)}\big(s^*(v_i'),\b_{-i}^{(t)}\big)$ and $v_i>v_i'$, the above inequality implies
\[
    \sum_{t\in[T]}\big(x_i^{(t)}\big(s^*(v_i'),\b_{-i}^{(t)}\big)-x_i^{(t)}\big(s^*(v_i),\b_{-i}^{(t)}\big)\big)\cdot v_i'<\sum_{t\in[T]}p^{(t)}_i\big(s^*(v_i'),\b_{-i}^{(t)}\big)-p^{(t)}_i\big(s^*(v_i),\b_{-i}^{(t)}\big).
\]
By rearranging, we obtain that
\[
    \sum_{t\in[T]}x_i^{(t)}\big(s^*(v_i'),\b_{-i}^{(t)}\big)\cdot v_i'-p^{(t)}_i\big(s^*(v_i'),\b_{-i}^{(t)}\big)<\sum_{t\in[T]}x_i^{(t)}\big(s^*(v_i),\b_{-i}^{(t)}\big)\cdot v_i'-p^{(t)}_i\big(s^*(v_i),\b_{-i}^{(t)}\big),
\]
which contradicts our choice of $s^*(v_i')$.
\end{proof}

\subsection{Proof of Lemma~\ref{lem:regret_quantile_strategy}}\label{section:proof_of_lemma_regret_quantile_strategy}
\begin{proof}[Proof of Lemma~\ref{lem:regret_quantile_strategy}]
Given bidder $i$'s bidding strategies $s_{i,\bm{\pi}^{(t)}}$ for all $t\in[T]$, we derive that
\begin{align*}
    \Reg_i&=\Reg_i^{\dagger} &&\text{(By Lemma~\ref{lem:monotone_strategies_suffice})}\\
    &=\max_{s\in\S^{\dagger}} \sum_{t\in[T]} u^{(t)}_{i}(s) - u^{(t)}_{i}(s_{i,\bm{\pi}^{(t)}}) &&\text{(By Eq.~\eqref{eq:regret_dagger})}\\
    &=\max_{s\in\S^{\dagger}} \sum_{t\in[T]} u^{(t)}_{i}(s) - q_i^{(t)}(\bm{\pi}^{(t)}) &&\text{(By Eq.~\eqref{eq:quantile_utility_function})}.
\end{align*}
Therefore, it suffices to prove that $\max_{s\in\S^{\dagger}} \sum_{t\in[T]} u^{(t)}_{i}(s)=\max_{\bm{\pi}\in\Delta([K+1])} \sum_{t\in[T]}q_i^{(t)}(\bm{\pi})$.

To this end, we let $s^*:=\argmax_{s\in\S^{\dagger}} \sum_{t\in[T]} u^{(t)}_{i}(s)$. Moreover, we define
\[
    v_j^*:=
    \begin{cases}
        0 & \textrm{if $s^*(v)>\frac{j-1}{K},\,\,\forall\,v\in[0,1]$}\\
        \sup \{v\in[0,1] \mid s(v)\le\tfrac{j-1}{K}\} & \textrm{if otherwise},
    \end{cases}
    ,\,\,\forall\,j\in[K+1].
\]
Intuitively, $v_j^*$ is the value threshold at which the monotone bidding strategy $s^*$ raises the bid from $\frac{j-1}{K}$ to $\frac{j}{K}$. For completeness, we let $v_0^*:=0$. We construct a quantile strategy $\bm{\pi}^*$ as follows,
\[
    \pi^*_k:=F_i(v_k^*)-F_i(v_{k-1}^*),\,\,\forall\,k\in[K+1],
\]
which implies that $\sum_{k=1}^j\pi^*_k=F_i(v_j^*)$ for all $j\in[K+1]$. Then, it follows by Eq.~\eqref{eq:s_i_pi} that
\begin{equation}\label{eq:s_i_pi_star}
    s_{i,\bm{\pi}^*}(v)=
    \begin{cases}
    0 & \textrm{if } v\in[0,v_1^*]  \\
    \frac{j-1}{K} & \textrm{if } v\in(v_{j-1}^*,v_j^*] \textrm{ for } j\in\{2,\dots,K+1\}.
    \end{cases}
\end{equation}

Now we show that $u^{(t)}_{i}(s^*)=q_i^{(t)}(\bm{\pi}^*)$ for all $t\in[T]$. Intuitively, this holds because the monotone bidding strategy $s_{i,\bm{\pi}^*}(v)$ coincides with $s^*$, except possibly at the value thresholds $v_j^*$. Formally, we derive that
\begin{align*}
    u^{(t)}_{i}(s^*)&=\E_{v_i\sim\D_i}\big[x_i^{(t)}\big(s^*(v_i),\b_{-i}^{(t)}\big)\cdot v_i-p^{(t)}_i\big(s^*(v_i),\b_{-i}^{(t)}\big)\big] &&\text{(By Eq.~\eqref{eq:utility_function})}\\
    &=\sum_{j\in[K+1]}\int_{v_{j-1}^*}^{v_j^*} x_i^{(t)}\big(\tfrac{j-1}{K},\b_{-i}^{(t)}\big)\cdot v_i-p^{(t)}_i\big(\tfrac{j-1}{K},\b_{-i}^{(t)}\big)dF_i(v_i) &&\text{(By definition of $v_j^*$'s)}\\
    &=\E_{v_i\sim\D_i}\big[x_i^{(t)}\big(s_{i,\bm{\pi}^*}(v_i),\b_{-i}^{(t)}\big)\cdot v_i-p^{(t)}_i\big(s_{i,\bm{\pi}^*}(v_i),\b_{-i}^{(t)}\big)\big] &&\text{(By Eq.~\eqref{eq:s_i_pi_star})}\\
    &=u^{(t)}_{i}(s_{i,\bm{\pi}^*})&&\text{(By Eq.~\eqref{eq:utility_function})}\\
    &=q_i^{(t)}(\bm{\pi}^*)&&\text{(By Eq.~\eqref{eq:quantile_utility_function})}.
\end{align*}
Hence, we have that
\[
    \max_{s\in\S^{\dagger}} \sum_{t\in[T]} u^{(t)}_{i}(s)=\sum_{t\in[T]} u^{(t)}_{i}(s^*)=\sum_{t\in[T]}q_i^{(t)}(\bm{\pi}^*)\le\max_{\bm{\pi}\in\Delta([K+1])} \sum_{t\in[T]}q_i^{(t)}(\bm{\pi}).
\]
Moreover, we notice that by Eq.~\eqref{eq:quantile_utility_function} and the fact that $s_{i,\bm{\pi}}\in\S^{\dagger}$,
\[
    \max_{\bm{\pi}\in\Delta([K+1])} \sum_{t\in[T]}q_i^{(t)}(\bm{\pi})=\max_{\bm{\pi}\in\Delta([K+1])} \sum_{t\in[T]}u^{(t)}_{i}(s_{i,\bm{\pi}})\le\max_{s\in\S^{\dagger}} \sum_{t\in[T]} u^{(t)}_{i}(s).
\]
It follows that $\max_{s\in\S^{\dagger}} \sum_{t\in[T]} u^{(t)}_{i}(s)=\max_{\bm{\pi}\in\Delta([K+1])} \sum_{t\in[T]}q_i^{(t)}(\bm{\pi})$, which finishes the proof.
\end{proof}

\subsection{Proof of Lemma~\ref{lem:quantile_utility_function_is_concave}}\label{section:proof_of_lemma_quantile_utility_function_is_concave}
\begin{proof}[Proof of Lemma~\ref{lem:quantile_utility_function_is_concave}]
We first expand the function $q_i^{(t)}(\bm{\pi})$ as follows,
\begin{align}\label{eq:expand_quantile_utility_function}
    q_i^{(t)}(\bm{\pi})&=u^{(t)}_{i}(s_{i,\bm{\pi}}) &&\text{(By Eq.~\eqref{eq:quantile_utility_function})}\nonumber\\
    &=\E_{v_i\sim\D_i}\big[x_i^{(t)}\big(s_{i,\bm{\pi}}(v_i),\b_{-i}^{(t)}\big)\cdot v_i-p^{(t)}_i\big(s_{i,\bm{\pi}}(v_i),\b_{-i}^{(t)}\big)\big] &&\text{(By Eq.~\eqref{eq:utility_function})}\nonumber\\
    &=\sum_{j\in[K]}\int_{F_i^{-1}(\sum_{\ell=1}^{j-1}\pi_{\ell})}^{F_i^{-1}(\sum_{\ell=1}^{j}\pi_{\ell})} x_i^{(t)}\big(\tfrac{j-1}{K},\b_{-i}^{(t)}\big)\cdot v_i-p^{(t)}_i\big(\tfrac{j-1}{K},\b_{-i}^{(t)}\big)dF_i(v_i) \nonumber\\
    &\quad+\int_{F_i^{-1}(\sum_{\ell=1}^{K}\pi_{\ell})}^{1} x_i^{(t)}\big(1,\b_{-i}^{(t)}\big)\cdot v_i-p^{(t)}_i\big(1,\b_{-i}^{(t)}\big)dF_i(v_i) &&\text{(By Eq.~\eqref{eq:s_i_pi})}\nonumber\\
    &=\sum_{j\in[K]}\int_{\sum_{\ell=1}^{j-1}\pi_{\ell}}^{\sum_{\ell=1}^{j}\pi_{\ell}} x_i^{(t)}\big(\tfrac{j-1}{K},\b_{-i}^{(t)}\big)\cdot F_i^{-1}(z)-p^{(t)}_i\big(\tfrac{j-1}{K},\b_{-i}^{(t)}\big)dz\nonumber\\
    &\quad+\int_{\sum_{\ell=1}^{K}\pi_{\ell}}^{1} x_i^{(t)}\big(1,\b_{-i}^{(t)}\big)\cdot F_i^{-1}(z)-p^{(t)}_i\big(1,\b_{-i}^{(t)}\big)dz.
\end{align}

Then, we define the partial-sum functions $\tau_j(\bm{\pi}):=\sum_{\ell=1}^j\pi_{\ell}$ for all $j\in[K]$, and we define an intermediate function $\sigma_i^{(t)}:[0,1]^{K}\to[-1,1]$ as follows\footnote{We note that $\sigma_i^{(t)}$ is essentially a generalization of bidders' utility functions in~\citet[Eq.~(3)]{KSS24} from first-price auctions to arbitrary auctions.},
\[
    \sigma_i^{(t)}(y_1,\dots,y_K):=\sum_{j\in[K+1]}\int_{y_{j-1}}^{y_j}x_i^{(t)}\big(\tfrac{j-1}{K},\b_{-i}^{(t)}\big)\cdot F_i^{-1}(z)-p^{(t)}_i\big(\tfrac{j-1}{K},\b_{-i}^{(t)}\big)dz,
\]
where we let $y_0:=0$ and $y_{K+1}:=1$. We notice that $q_i^{(t)}(\bm{\pi})=\sigma_i^{(t)}(\tau_1(\bm{\pi}),\dots,\tau_K(\bm{\pi}))$ by Eq.~\eqref{eq:expand_quantile_utility_function}, and moreover, for each $j\in[K]$, $\tau_j$ is a linear function of $\bm{\pi}$.

\subsubsection*{Concavity of $q_i^{(t)}$}
Since linear transformations preserve concavity~\citep[Section 3.2]{BV04}, to show that $q_i^{(t)}$ is concave, it suffices to prove that $\sigma_i^{(t)}$ is concave. To this end, we observe that the partial derivatives of $\sigma_i^{(t)}$ with respect to $y_j$, for all $j\in[K]$, are given by
\begin{align}\label{eq:sigma_partial_gradient}
\frac{\partial \sigma_i^{(t)}}{\partial y_j}&=\big(x_i^{(t)}\big(\tfrac{j-1}{K},\b_{-i}^{(t)}\big)-x_i^{(t)}\big(\tfrac{j}{K},\b_{-i}^{(t)}\big)\big)\cdot F_i^{-1}(y_j)-\big(p^{(t)}_i\big(\tfrac{j-1}{K},\b_{-i}^{(t)}\big)-p^{(t)}_i\big(\tfrac{j}{K},\b_{-i}^{(t)}\big)\big),\,\,\forall\,j\in[K].
\end{align}
We notice that for all $j\in[K]$, the partial derivative $\frac{\partial \sigma_i^{(t)}}{\partial y_j}$ depends only on $y_j$. Thus, to prove that $\sigma_i^{(t)}$ is concave, it suffices to show that for all $j\in[K]$, $\frac{\partial \sigma_i^{(t)}}{\partial y_j}$ is a non-increasing function of $y_j$. We observe that $F_i^{-1}(y_j)$ is non-decreasing in $y_j$ since $F_i$ is a CDF, and that $x_i^{(t)}\big(\tfrac{j-1}{K},\b_{-i}^{(t)}\big)-x_i^{(t)}\big(\tfrac{j}{K},\b_{-i}^{(t)}\big)$ is non-positive because $x_i^{(t)}$ is non-decreasing by Assumption~\ref{assumption:auction_format}-\ref{assumption:allocation_monotonicity}. It follows that $\frac{\partial \sigma_i^{(t)}}{\partial y_j}$ is a non-increasing function of $y_j$ for all $j\in[K]$, which implies that $\sigma_i^{(t)}$, and hence $q_i^{(t)}$, is concave.

\subsubsection*{Gradient of $q_i^{(t)}$}
Next, we note that $\frac{\partial q_i^{(t)}}{\partial \pi_{K+1}}=0$ since $q_i^{(t)}$ is independent of $\pi_{K+1}$. Moreover, since $q_i^{(t)}(\bm{\pi})=\sigma_i^{(t)}(\tau_1(\bm{\pi}),\dots,\tau_K(\bm{\pi}))$, for each $k\in[K]$, we can compute the partial derivative of $q_i^{(t)}$ with respect to $\pi_k$ as follows,
\begin{align*}
\frac{\partial q_i^{(t)}}{\partial \pi_k}&=\sum_{j=1}^{K}\frac{\partial \sigma_i^{(t)}}{\partial \tau_j}\cdot\frac{\partial \tau_j}{\partial \pi_k}=\sum_{j=k}^{K}\frac{\partial \sigma_i^{(t)}}{\partial \tau_j}\\
&=\sum_{j=k}^{K}\big(x_i^{(t)}\big(\tfrac{j-1}{K},\b_{-i}^{(t)}\big)-x_i^{(t)}\big(\tfrac{j}{K},\b_{-i}^{(t)}\big)\big)\cdot F_i^{-1}(\tau_j(\bm{\pi}))-\big(p^{(t)}_i\big(\tfrac{j-1}{K},\b_{-i}^{(t)}\big)-p^{(t)}_i\big(\tfrac{j}{K},\b_{-i}^{(t)}\big)\big)\\
&\qquad\text{(By Eq.~\eqref{eq:sigma_partial_gradient})}\\
&=\bigg(\sum_{j=k}^{K}\big(x_i^{(t)}\big(\tfrac{j-1}{K},\b_{-i}^{(t)}\big)-x_i^{(t)}\big(\tfrac{j}{K},\b_{-i}^{(t)}\big)\big)\cdot F_i^{-1}(\tau_j(\bm{\pi}))\bigg)+p_i^{(t)}\big(1,\b_{-i}^{(t)}\big)-p^{(t)}_i\big(\tfrac{k-1}{K},\b_{-i}^{(t)}\big)\\
&\qquad\text{(By a telescoping sum)},
\end{align*}
which establishes Eq.~\eqref{eq:gradient_of_quantile_utility_function}. To bound $\frac{\partial q_i^{(t)}}{\partial \pi_k}$ for any $k\in[K]$, we notice that each term $x_i^{(t)}\big(\tfrac{j-1}{K},\b_{-i}^{(t)}\big)-x_i^{(t)}\big(\tfrac{j}{K},\b_{-i}^{(t)}\big)$ is non-positive (by Assumption~\ref{assumption:auction_format}-\ref{assumption:allocation_monotonicity}) and $F_i^{-1}(\tau_j(\bm{\pi}))\in[0,1]$. Hence, we have that
\begin{align*}
    0\ge\sum_{j=k}^{K}\big(x_i^{(t)}\big(\tfrac{j-1}{K},\b_{-i}^{(t)}\big)-x_i^{(t)}\big(\tfrac{j}{K},\b_{-i}^{(t)}\big)\big)\cdot F_i^{-1}(\tau_j(\bm{\pi}))&\ge \sum_{j=k}^{K} x_i^{(t)}\big(\tfrac{j-1}{K},\b_{-i}^{(t)}\big)-x_i^{(t)}\big(\tfrac{j}{K},\b_{-i}^{(t)}\big)\\
    &=x_i^{(t)}\big(\tfrac{k-1}{K},\b_{-i}^{(t)}\big)-x_i^{(t)}\big(1,\b_{-i}^{(t)}\big).
\end{align*}
Combining this with Eq.~\eqref{eq:gradient_of_quantile_utility_function}, we obtain that
\begin{align*}
    \frac{\partial q_i^{(t)}}{\partial \pi_k}&\le p_i^{(t)}\big(1,\b_{-i}^{(t)}\big)-p^{(t)}_i\big(\tfrac{k-1}{K},\b_{-i}^{(t)}\big) \le 1,\\
    \frac{\partial q_i^{(t)}}{\partial \pi_k}&\ge x_i^{(t)}\big(\tfrac{k-1}{K},\b_{-i}^{(t)}\big)-x_i^{(t)}\big(1,\b_{-i}^{(t)}\big)+p_i^{(t)}\big(1,\b_{-i}^{(t)}\big)-p^{(t)}_i\big(\tfrac{k-1}{K},\b_{-i}^{(t)}\big)\\
    &\ge -2,
\end{align*}
which completes the proof of the lemma.
\end{proof}

\subsection{Proof of Claim~\ref{claim:auxiliary_auction_payment}}\label{section:proof_of_claim_auxiliary_auction_payment}
\begin{proof}[Proof of Claim~\ref{claim:auxiliary_auction_payment}]
Since bidders' value vector $\v^{(t)}$ at round $t$ is sampled from $\D$ independently of $\tilde{\M}^{(t)}$, we can write $\Rev_{\tilde{\M}^{(t)}}(\D)=\E_{\v^{(t)}\sim\D}\big[\sum_{i\in[n]}\tilde{p}_i^{(t)}(\v^{(t)})\big]$. We express $\tilde{p}_i^{(t)}(\v^{(t)})$ as follows,
\begin{align}\label{eq:auxiliary_auction_payment_expression}
\tilde{p}_i^{(t)}(\v^{(t)})&=\tilde{x}_i^{(t)}\big(v_i^{(t)},\v_{-i}^{(t)}\big)v_i^{(t)} - \int_{0}^{v_i^{(t)}} \tilde{x}_i^{(t)}\big(z,\v_{-i}^{(t)}\big)dz &&\text{(By Eq.~\eqref{eq:auxiliary_auction_payment})}\nonumber\\
&=x_i^{(t)}\big(s_1^{(t)}(v_1^{(t)}),\dots,s_n^{(t)}(v_n^{(t)})\big)v_i^{(t)}\nonumber\\
&\quad-\int_{0}^{v_i^{(t)}} x_i^{(t)}\big(s_1^{(t)}(v_1^{(t)}),\dots,s_i^{(t)}(z),\dots,s_n^{(t)}(v_n^{(t)})\big)dz &&\text{(By Eq.~\eqref{eq:auxiliary_auction_allocation})}\nonumber\\
&=x_i^{(t)}\big(s_i^{(t)}(v_i^{(t)}),\b_{-i}^{(t)}\big)v_i^{(t)}-\int_{0}^{v_i^{(t)}} x_i^{(t)}\big(s_i^{(t)}(z),\b_{-i}^{(t)}\big)dz.
\end{align}
Recall that $s_i^{(t)}=s_{i,\bm{\pi}^{(i,t)}}$ in Algorithm~\ref{alg:meta}. For compactness, let $\Pi_0^{(i,t)}:=0$ and $\Pi_k^{(i,t)}:=\sum_{\ell=1}^k\pi_{\ell}^{(i,t)}$ for all $k\in[K+1]$. By Eq.~\eqref{eq:s_i_pi}, for any $v\in\big(F_i^{-1}(\Pi_{k-1}^{(i,t)}),F_i^{-1}(\Pi_k^{(i,t)})\big]$ with $k\in[K+1]$, it holds that $s_i^{(t)}(v)=s_{i,\bm{\pi}^{(i,t)}}(v)=\frac{k-1}{K}$. Hence, by taking the expectation of both sides of Eq.~\eqref{eq:auxiliary_auction_payment_expression} over $v_i^{(t)}$ conditioned on $\v_{-i}^{(t)}$ (which determines $\b_{-i}^{(t)}$), we obtain that
\begin{align}\label{eq:auxiliary_auction_expected_payment}
    &\E_{v_i^{(t)}\sim\D_i}\big[\tilde{p}_i^{(t)}(\v^{(t)})\mid \v_{-i}^{(t)}\big]\nonumber\\
    =&\sum_{k\in[K+1]}\int_{F_i^{-1}(\Pi_{k-1}^{(i,t)})}^{F_i^{-1}(\Pi_k^{(i,t)})}\Big(x_i^{(t)}\Big(\frac{k-1}{K},\b_{-i}^{(t)}\Big)v_i^{(t)}-\int_{0}^{v_i^{(t)}} x_i^{(t)}\big(s_i^{(t)}(z),\b_{-i}^{(t)}\big)dz\Big)\,dF_i(v_i^{(t)}).
\end{align}
Similarly, for any $v_i^{(t)}\in\big(F_i^{-1}(\Pi_{k-1}^{(i,t)}),F_i^{-1}(\Pi_k^{(i,t)})\big]$, we can decompose $\int_{0}^{v_i^{(t)}} x_i^{(t)}\big(s_i^{(t)}(z),\b_{-i}^{(t)}\big)dz$ as
\begin{align*}
    &\,\int_{0}^{v_i^{(t)}} x_i^{(t)}\big(s_i^{(t)}(z),\b_{-i}^{(t)}\big)dz\\
    =&\sum_{j\in[k-1]}\int_{F_i^{-1}(\Pi_{j-1}^{(i,t)})}^{F_i^{-1}(\Pi_{j}^{(i,t)})}x_i^{(t)}\Big(\frac{j-1}{K},\b_{-i}^{(t)}\Big)dz+\int_{F_i^{-1}(\Pi_{k-1}^{(i,t)})}^{v_i^{(t)}}x_i^{(t)}\Big(\frac{k-1}{K},\b_{-i}^{(t)}\Big)dz\\
    =&\sum_{j\in[k-1]}\big(F_i^{-1}(\Pi_{j}^{(i,t)})-F_i^{-1}(\Pi_{j-1}^{(i,t)})\big)\cdot x_i^{(t)}\Big(\frac{j-1}{K},\b_{-i}^{(t)}\Big)+\big(v_i^{(t)}-F_i^{-1}(\Pi_{k-1}^{(i,t)})\big)\cdot x_i^{(t)}\Big(\frac{k-1}{K},\b_{-i}^{(t)}\Big)\\
    =&\sum_{j\in[k-1]}\Big(x_i^{(t)}\Big(\frac{j-1}{K},\b_{-i}^{(t)}\Big)-x_i^{(t)}\Big(\frac{j}{K},\b_{-i}^{(t)}\Big)\Big)\cdot F_i^{-1}(\Pi_{j}^{(i,t)})+x_i^{(t)}\Big(\frac{k-1}{K},\b_{-i}^{(t)}\Big)v_i^{(t)}\\
    &\quad\text{(By $F_i^{-1}(\Pi_{0}^{(i,t)})=F_i^{-1}(0)=0$ and rearranging)}.
\end{align*}
By substituting this into Eq.~\eqref{eq:auxiliary_auction_expected_payment}, we derive that
\begin{align*}
    &\E_{v_i^{(t)}\sim\D_i}\big[\tilde{p}_i^{(t)}(\v^{(t)})\mid \v_{-i}^{(t)}\big]\\
    =&\sum_{k\in[K+1]}\int_{F_i^{-1}(\Pi_{k-1}^{(i,t)})}^{F_i^{-1}(\Pi_k^{(i,t)})}\sum_{j\in[k-1]}\Big(x_i^{(t)}\Big(\frac{j}{K},\b_{-i}^{(t)}\Big)-x_i^{(t)}\Big(\frac{j-1}{K},\b_{-i}^{(t)}\Big)\Big)\cdot F_i^{-1}(\Pi_{j}^{(i,t)})dF_i(v_i^{(t)})\\
    =&\sum_{k\in[K+1]}\int_{\Pi_{k-1}^{(i,t)}}^{\Pi_k^{(i,t)}}\sum_{j\in[k-1]}\Big(x_i^{(t)}\Big(\frac{j}{K},\b_{-i}^{(t)}\Big)-x_i^{(t)}\Big(\frac{j-1}{K},\b_{-i}^{(t)}\Big)\Big)\cdot F_i^{-1}(\Pi_{j}^{(i,t)})dz\\
    =&\sum_{k\in[K+1]}\pi_k^{(i,t)}\cdot\sum_{j\in[k-1]}\Big(x_i^{(t)}\Big(\frac{j}{K},\b_{-i}^{(t)}\Big)-x_i^{(t)}\Big(\frac{j-1}{K},\b_{-i}^{(t)}\Big)\Big)\cdot F_i^{-1}(\Pi_{j}^{(i,t)})\\
    =&\sum_{j=1}^K\sum_{k=j+1}^{K+1}\pi_k^{(i,t)}\cdot\Big(x_i^{(t)}\Big(\frac{j}{K},\b_{-i}^{(t)}\Big)-x_i^{(t)}\Big(\frac{j-1}{K},\b_{-i}^{(t)}\Big)\Big)\cdot F_i^{-1}(\Pi_{j}^{(i,t)})\qquad\,\,\,\text{(By exchanging the sums)}.
\end{align*}
By taking the expectation of both sides over $\v_{-i}^{(t)}$ (which determines $\b_{-i}^{(t)}$), we have that
\[
    \E_{\v^{(t)}\sim\D}\big[\tilde{p}_i^{(t)}(\v^{(t)})\big]=\E\!_{\v_{-i}^{(t)}}\Big[\sum_{j=1}^K\sum_{k=j+1}^{K+1}\pi_k^{(i,t)}\cdot \Big(x_i^{(t)}\Big(\frac{j}{K},\b_{-i}^{(t)}\Big)-x_i^{(t)}\Big(\frac{j-1}{K},\b_{-i}^{(t)}\Big)\Big)\cdot F_i^{-1}(\Pi_{j}^{(i,t)})\Big],
\]
which implies the claim because $\Rev_{\tilde{\M}^{(t)}}(\D)=\sum_{i\in[n]}\E_{\v^{(t)}\sim\D}\big[\tilde{p}_i^{(t)}(\v^{(t)})\big]$.
\end{proof}

\subsection{Proof of Proposition~\ref{prop:lower_bound}}\label{section:proof_of_prop_lower_bound}
\begin{proof}[Proof of Proposition~\ref{prop:lower_bound}]
Suppose that there is only one bidder in the repeated auction game, and the bidder's value distribution $\D_1$ is the uniform distribution over $\big[1-\frac{1}{T},1\big]$\footnote{The proof applies to any continuous distribution with sufficiently large probability mass on $\big[1-\frac{1}{T},1\big]$.}. Given any algorithm $\A$ that guarantees regret $\Reg_{\A}$ for the bidder in the repeated auction game, we show how to convert $\A$ into an algorithm $\A'$ that guarantees regret $\Reg_{\A}+O(1)$ in the standard online learning setting (Definition~\ref{def:no_regret_learning_algorithm}), but with reward vectors in $\{0,1\}^K$. This will establish the proposition, since the worst-case regret of any algorithm in the standard online learning setting (with reward vectors in $\{0,1\}^K$) is at least $\Omega(\sqrt{T\log(K)})$~\citep[Theorem 3.7]{CG06}.

The algorithm $\A'$ operates as follows: At each round $t\in[T]$, it first invokes algorithm $\A$, which returns a bidding strategy $s^{(t)}$. Let $\bm{\pi}^{(t)}\in\Delta([K+1])$ be the distribution such that
\[
    \pi_i^{(t)}:=\Pr_{v\sim\D_1}\big[s^{(t)}(v)=\tfrac{i-1}{K}\big],\,\,\forall\,i\in[K+1].
\]
That is, $\pi_i^{(t)}$ is the probability mass of all value types that bid $\frac{i-1}{K}$ under the bidding strategy $s^{(t)}$. Algorithm $\A'$ then chooses each action $i\in\{2,\dots,K\}$ with probability $\pi_{i+1}^{(t)}$, and chooses action $1$ with probability $\pi_1^{(t)}+\pi_2^{(t)}$, after which a reward vector $\r^{(t)}\in\{0,1\}^K$ is revealed. Algorithm $\A'$, in turn, reveals to $\A$ an auction format $\M^{(t)}$ that allocates the item to the bidder with probability $1$ if her bid is non-zero and with probability $0$ otherwise. The payment function $p^{(t)}:B\to[0,1]$ of $\M^{(t)}$ is given by $p^{(t)}(0)=0$ and
\[
    p^{(t)}\big(\tfrac{i}{K}\big):=1-r^{(t)}_i,\,\,\forall\,i\in[K].
\]

We observe that the regret of algorithm $\A'$ (in the standard online setting) is
\begin{equation}\label{eq:lower_bound_proof_regret_A'}
    \Reg_{\A'}=\sum_{t\in[T]}\Big(r_{i^*}^{(t)}-\sum_{i\in[K]}\pi_{i+1}^{(t)}r_i^{(t)}-\pi_1^{(t)}r_1^{(t)}\Big)\le\sum_{t\in[T]}\Big(r_{i^*}^{(t)}-\sum_{i\in[K]}\pi_{i+1}^{(t)}r_i^{(t)}\Big),
\end{equation}
where $i^*:=\argmax_{i\in[K]}\sum_{t\in[T]}r_{i}^{(t)}$. On the other hand, for each $t\in[T]$, let $v^{(t)}$ denote the bidder's value at round $t$, and notice that the bidder's utility obtained by bidding $0$ is zero, and for each $i\in[K]$, the bidder's utility achieved by bidding $\frac{i}{K}$ at round $t$ is
\begin{equation}\label{eq:lower_bound_proof_utility_round_t}
    v^{(t)}-p^{(t)}\big(\tfrac{i}{K}\big)=r_{i}^{(t)}-O\big(\tfrac{1}{T}\big),
\end{equation}
since $v^{(t)}\in[1-\frac{1}{T},1]$ and $p^{(t)}\big(\frac{i}{K}\big)=1-r^{(t)}_i$. Hence, the bidder's cumulative utility by bidding $\frac{i^*}{K}$ over all $T$ rounds is $\sum_{t\in[T]}\big(r_{i^*}^{(t)}-O\big(\frac{1}{T}\big)\big)$. Since the best bidding strategy in hindsight is no worse than always bidding $\frac{i^*}{K}$, the regret of the algorithm $\A$ (in the repeated auction game) is at least
\begin{align*}
    \Reg_{\A}&\ge\sum_{t\in[T]}\Big(r_{i^*}^{(t)}-O\big(\tfrac{1}{T}\big)-\sum_{i\in[K]}\pi_{i+1}^{(t)}\cdot\big(v^{(t)}-p^{(t)}\big(\tfrac{i}{K}\big)\big)\Big)\\
    &=\sum_{t\in[T]}\Big(r_{i^*}^{(t)}-O\big(\tfrac{1}{T}\big)-\sum_{i\in[K]}\pi_{i+1}^{(t)}\cdot\big(r_i^{(t)}-O\big(\tfrac{1}{T}\big)\big)\Big) &&\text{(By Eq.~\eqref{eq:lower_bound_proof_utility_round_t})}\\
    &\ge\Reg_{\A'}-O(1) &&\text{(By Eq.~\eqref{eq:lower_bound_proof_regret_A'})},
\end{align*}
which implies that $\Reg_{\A'}\le\Reg_{\A}+O(1)$.
\end{proof}

\begin{algorithm}[ht]
\SetAlgoLined
\SetKwInOut{Input}{Input}
\SetKwInOut{Output}{Output}
\Input{Step size $\eta>0$}
\SetAlgorithmName{Algorithm}~~
    $\bm{\pi}^{(1)}\gets\big(\frac{1}{K+1},\dots,\frac{1}{K+1}\big)$\;
    \For{$t=1,\dots,T$}{
        Play an action $i\in[K+1]$ sampled from the distribution $\bm{\pi}^{(t)}$\;
        Observe the reward vector $\r^{(t)}$\;
        $\bm{\pi}^{(t+1)}\gets\argmin_{\bm{\pi}\in\Delta([K+1])}||\bm{\pi}-(\bm{\pi}^{(t)}+\eta\cdot\r^{(t)})||_2$\;
    }
    \caption{\textsc{Agile Online Gradient Ascent}}
    \label{alg:agile_OGD}
\end{algorithm}
\section{A high-swap-regret instance for agile OGD}\label{section:swap_regret}
We note that there is a growing body of work on stronger notions of regret that agile OGD minimizes (e.g.,~\citet{ahunbay24,AB25,CDLWZ25}), and it is known that agile OGD does not minimize swap regret (see e.g.,~\citet{AB25}). However, we were unable to find in the literature an explicitly written example showing that agile OGD with a fixed step size can incur $\Omega(T)$ swap regret. For completeness, we formalize one such example here and hope that it may be useful for future research. We first define swap regret and present agile OGD (Algorithm~\ref{alg:agile_OGD})—more precisely, agile online gradient ascent, since we work with reward vectors rather than loss vectors—within the standard online learning setting (Definition~\ref{def:no_regret_learning_algorithm}).
\begin{definition}
Following the setup in Definition~\ref{def:no_regret_learning_algorithm}, algorithm $\A$'s swap regret is
\[
    \SwapReg_{\A}:=\max_{\phi:[K+1]\to[K+1]} \sum_{t\in[T]} \Big(\sum_{k\in[K+1]}\pi_k^{(t)}r_{\phi(k)}^{(t)} - \pi_k^{(t)}r_k^{(t)}\Big).
\]
\end{definition}

Now we describe the instance and show that Algorithm~\ref{alg:agile_OGD} incurs high swap regret on it.

\begin{example}\label{ex:swap_regret_instance}
Assume w.l.o.g.~that $\frac{1}{2\eta}$ and $\frac{\eta}{300}\cdot T$ are both integers. Let $\alpha:=\frac{100}{\eta}$ and $\beta:=\frac{\eta}{300}\cdot T$. Suppose that $K=2$, and therefore the action set $[K+1]$ contains three actions. The reward vectors arrive in $\beta$ identical batches. Each batch $i\in[\beta]$ has three phases, each containing $\alpha$ reward vectors, which are defined in Table~\ref{table:batch}.

\begin{table}[ht]
\begin{center}
\begin{tabular}{@{}l | c c c}
 &
 $r_1^{(t)}$ &
 $r_2^{(t)}$ & 
 $r_3^{(t)}$ \\
\hline
 Phase 1: $t\in\{(i-1)\cdot3\alpha+1,\dots,(i-1)\cdot3\alpha+\alpha\}$ &
  1 & 1 & 0\\
\hline
 Phase 2: $t\in\{(i-1)\cdot3\alpha+\alpha+1,\dots,(i-1)\cdot3\alpha+2\alpha\}$ &
  0 & 1 & 0\\
\hline
 Phase 3: $t\in\{(i-1)\cdot3\alpha+2\alpha+1,\dots,i\cdot3\alpha\}$ &
  0 & 0 & 1
\end{tabular}
\caption{Reward vectors in the $i$-th batch for each $i\in[\beta]$}\label{table:batch}
\end{center}
\end{table}
\end{example}

\begin{proposition}\label{prop:swap_regret}
Algorithm~\ref{alg:agile_OGD} incurs $\Omega(T)$ swap regret on Example~\ref{ex:swap_regret_instance}.
\end{proposition}
\begin{proof}
Consider the mapping $\phi:[3]\to[3]$ defined by $\phi(1)=2$, $\phi(2)=2$, and $\phi(3)=3$.
Then,
\[
\sum_{k=1}^3 \pi_k^{(t)}r_{\phi(k)}^{(t)}
-\sum_{k=1}^3 \pi_k^{(t)}r_k^{(t)}
= \pi_1^{(t)}\cdot\big(r_2^{(t)}-r_1^{(t)}\big),
\]
so the swap regret of Algorithm~\ref{alg:agile_OGD}, which we denote by $\SwapReg_{AOGD}$, is at least
\[
\sum_{t=1}^T \pi_1^{(t)}\cdot\big(r_2^{(t)}-r_1^{(t)}\big).
\]
From Table~\ref{table:batch}, only Phase~2 contributes to this sum, because
$r_2^{(t)}-r_1^{(t)}$ equals 1 in Phase~2 and equals $0$ in Phases~1 and~3. Hence, we have
\[
\SwapReg_{AOGD} \;\ge\; \sum_{t:\text{ Phase 2}} \pi_1^{(t)}.
\]

We analyze one batch and then argue that all batches behave identically.

\subsubsection*{Phase 1: $\r^{(t)}=(1,1,0)$}
Let $\bm{\pi}^{(t)}=\big(\pi_1^{(t)},\pi_2^{(t)},\pi_3^{(t)}\big)$ and suppose it satisfies
$\bm{\pi}^{(t)}=(x_t,x_t,1-2x_t)$ for some $x_t\in\big[0,\frac{1}{2}\big]$.
We update
\[
\bm{y}^{(t)}=\bm{\pi}^{(t)}+\eta(1,1,0)=(x_t+\eta,x_t+\eta,1-2x_t),
\]
and project $\bm{y}^{(t)}$ onto the simplex $\Delta([3])$ in $\ell_2$.
By symmetry, the projection has the form
\[
\bm{\pi}^{(t+1)} = (a,\, a,\, 1-2a)
\]
for some $a \in \big[0,\frac{1}{2}\big]$ that minimizes the squared distance
\[
\bigl\| (a,a,1-2a) - (x_t+\eta,\, x_t+\eta,\, 1-2x_t) \bigr\|_2^2
    = 2(a - x_t - \eta)^2 + 4(a - x_t)^2.
\]
Taking the derivative and setting it to zero gives
\[
12a - 12x_t - 4\eta = 0
\quad\Longrightarrow\quad
a = x_t + \frac{\eta}{3}.
\]
As long as the third component remains non-negative, i.e.,
\[
1 - 2a = 1 - 2x_t - \frac{2\eta}{3} \ge 0
    \quad\Longleftrightarrow\quad
    x_t \le \frac{1}{2} - \frac{\eta}{3},
\]
the projection update is
\[
\bm{\pi}^{(t+1)}
    = \bigl(x_t+\tfrac{\eta}{3},\, x_t+\tfrac{\eta}{3},\, 1 - 2x_t - \tfrac{2\eta}{3}\bigr).
\]

Starting from $\bm{\pi}^{(t)} = (0,0,1)$ at the beginning of Phase~1 of each batch 
(in the analysis of Phase 3, we will explain why every batch begins with this state, except the very first batch), we have $x_t=0$ and hence
\[
\pi_1^{(t+m)} = \pi_2^{(t+m)} = \frac{\eta}{3}\cdot m\,\, \textrm{ and }\,\,
\pi_3^{(t+m)} = 1 - \frac{2\eta}{3}\cdot m,\,\,\forall\, m=1,\dots,\frac{3}{2\eta},
\]
Thus, $\pi_3^{(t)}$ decreases linearly and hits $0$ after $\frac{3}{2\eta}$ steps.
At that point,
\[
\left(\pi_1^{\left(t+\frac{3}{2\eta}\right)},\pi_2^{\left(t+\frac{3}{2\eta}\right)},\pi_3^{\left(t+\frac{3}{2\eta}\right)}\right) = \Bigl(\frac12,\frac12,0\Bigr).
\]
Since $\alpha = \frac{100}{\eta} > \frac{3}{2\eta}$, Phase~1 contains more than enough rounds to reach
$(\tfrac12,\tfrac12,0)$. Observe that this point is a fixed point under the update rule of Algorithm~\ref{alg:agile_OGD} when $\r^{(t)} = (1,1,0)$. Hence the state remains $(\tfrac12,\tfrac12,0)$ for the remainder 
of Phase~1.

Analogously, for the very first batch, the same recurrence applied to the initial distribution 
$(\tfrac13,\tfrac13,\tfrac13)$ reaches $(\tfrac12,\tfrac12,0)$ within $\alpha=\frac{100}{\eta}$ rounds.

\subsubsection*{Phase 2: $\r^{(t)}=(0,1,0)$}
At the start of Phase~2 in any batch, the iterate is $\bm{\pi}^{(t)}=(\tfrac12,\tfrac12,0)$.
During this phase, the state stays on the edge $\{(\pi_1,\pi_2,0):\pi_1+\pi_2=1\}$.
Suppose that $\pi_1^{(t)}=x_t$, and hence, $\pi_2^{(t)}=1-x_t$. We update
\[
\bm{y}^{(t)} = \bm{\pi}^{(t)} + \eta(0,1,0) = (x_t,\,1-x_t+\eta,\,0),
\]
and project onto the edge, i.e., we minimize over $a$:
\[
(a-x_t)^2 + ((1-a)-(1-x_t+\eta))^2
= (a-x_t)^2 + (a-(x_t-\eta))^2.
\]
Taking derivative w.r.t.~$a$ and setting to zero, we obtain
\[
2(a-x_t) + 2(a-(x_t-\eta)) = 0
\quad\Longrightarrow\quad a = x_t - \frac{\eta}{2}.
\]
Thus, as long as $x_t\ge \frac{\eta}{2}$, the update is
\[
x_{t+1} = x_t - \frac{\eta}{2}.
\]
Starting from $x_t=\tfrac12$, we obtain for the first $\frac{1}{\eta}$ steps of Phase~2
\[
x_{t+m} = \frac12 - \frac{\eta}{2}\cdot m,\,\,\forall\, m=1,\dots,\frac{1}{\eta},
\]
and after $\frac{1}{\eta}<\alpha$ steps, $x_{t+\frac{1}{\eta}}$ decreases to $0$, and the state becomes $(0,1,0)$
and stays there for the remainder of Phase~2.

Hence, this batch contributes at least
\[
\sum_{m=0}^{1/\eta-1} x_{t+m}
= \sum_{m=0}^{1/\eta-1} \left(\frac12 - \frac{\eta}{2}\cdot m\right)
= \frac{1}{4\eta} + \frac{1}{4}
= \Omega\!\left(\frac{1}{\eta}\right)
\]
to $\sum_{t:\text{ Phase 2}} \pi_1^{(t)}$, and therefore also to the swap regret.

\subsubsection*{Phase 3: $\r^{(t)}=(0,0,1)$}
At the start of Phase~3, the iterate is $\bm{\pi}^{(t)}=(0,1,0)$.
During this phase, the state stays on the edge $\{(0,\pi_2,\pi_3):\pi_2+\pi_3=1\}$.
Suppose that $\pi_2^{(t)}=x_t$, and hence, $\pi_3^{(t)}=1-x_t$. We update
\[
\bm{y}^{(t)} = \bm{\pi}^{(t)} + \eta(0,0,1) = (0,\,x_t,\,1-x_t+\eta),
\]
and project onto the edge. Analogously to Phase 2, this gives
\[
x_{t+1} = x_t - \frac{\eta}{2}.
\]
Starting from $x_t = 1$, after $\frac{2}{\eta}<\alpha$ steps, the iterate reaches $(0,0,1)$. Hence, each Phase~3 ends at $(0,0,1)$. Therefore, every new batch starts from $(0,0,1)$, and the above analysis of Phases~1 and~2 repeats identically in each batch.

\subsubsection*{Total swap regret}
We have shown that each batch contributes at least $\Omega\big(\frac{1}{\eta}\big)$
to $\sum_{t:\text{ Phase 2}} \pi_1^{(t)}$, and hence also to the swap regret. 
Since there are $\beta=\frac{\eta}{300}\cdot T$ batches,
\[
\SwapReg_{AOGD}
\;\ge\;
\beta \cdot \Omega\left(\frac{1}{\eta}\right)
= \Omega(T).
\]
\end{proof}

\end{document}